\documentclass[onefignum,onetabnum]{siamonline250211}%{siamart251104}
\usepackage{scrlfile}
\PreventPackageFromLoading{subfigure}
\usepackage[caption=false]{subfig}  % SIAM-compatible
\usepackage{lipsum}
\usepackage{amsfonts}
\usepackage{amssymb}
\usepackage{graphicx}
\usepackage{epstopdf}
\usepackage{algorithm}
\usepackage{algorithmic}
\usepackage[utf8]{inputenc}
\usepackage{amsmath}
\usepackage{latexsym}
\usepackage{subfigure}
\usepackage{stmaryrd}
\usepackage{caption}
\usepackage{xcolor}
\usepackage[colorinlistoftodos]{todonotes}
\usepackage{enumerate}
\usepackage{bbm}
\usepackage{bm}
\usepackage{xspace,comment}
\usepackage{hyperref}
\ifpdf
  \DeclareGraphicsExtensions{.eps,.pdf,.png,.jpg}
\else
  \DeclareGraphicsExtensions{.eps}
\fi

\usepackage{enumitem}
\setlist[enumerate]{leftmargin=.5in}
\setlist[itemize]{leftmargin=.5in}

\newsiamremark{remark}{Remark}
\newsiamremark{example}{Example}
\newsiamthm{claim}{Claim}
\newsiamthm{assumption}{Assumption}

\headers{Constrained Bayesian filters}{U. Erdogan, G. J. Lord and J. Miguez}

\title{On a class of constrained Bayesian filters and their numerical implementation in high-dimensional state-space models\thanks{Submitted to the editors on ...
\funding{This work has been partially supported by grant PID2024-158181NB-I00 NISA, funded by MCIN/AEI/10.13039/501100011033 and ERDF, and by Community of Madrid, under Grant IDEA-CM (TEC-
2024/COM-89).}}}

\author{Utku Erdogan\thanks{Eskisehir Technical University (\email{utkuerdog@gmail.com})}
\and Gabriel J. Lord\thanks{Radboud University (\email{gabriel.lord@ru.nl})}
\and Joaquin Miguez\thanks{Universidad Carlos III de Madrid (\email{joaquin.miguez@uc3m.es})}}

\usepackage{amsopn}

\newcommand{\QED}{\Box} 
    
\newcommand{\Real}{\mathbb{R}}  
\newcommand{\Ind}{\mathbbm{1}}

\newcommand\dtv[1]{\left\| #1 \right\|_{\rm tv}}

\newcommand{\mB}{{\mathcal B}}
\newcommand{\mP}{{\mathcal P}}

\newcommand{\mC}{{\mathcal C}}

\newcommand{\mS}{{\mathcal S}}

\newcommand{\mH}{{\mathcal H}}

\newcommand{\mF}{{\mathcal F}}

\newcommand{\mX}{{\mathcal X}}
\newcommand{\mY}{{\mathcal Y}}
\newcommand{\mN}{{\mathcal N}}

\newcommand{\mbN}{\mathbb{N}}

\newcommand{\mbP}{\mathbb{P}}   
\newcommand{\mbS}{\mathbb{S}} 
 
\newcommand{\mbR}{\mathbb{R}} 
\newcommand{\mbD}{\mathbb{D}} 

\newcommand{\sd}{{\sf d}}

\newcommand{\X}{\mathtt{X}}

\newcommand{\W}{\mathtt{W}}

\newcommand{\beq}{\begin{equation}}
\newcommand{\eeq}{\end{equation}}
\newcommand{\beqa}{\begin{eqnarray}}
\newcommand{\eeqa}{\end{eqnarray}}
\newcommand{\nn}{\nonumber}

\newcommand{\dy}{d_y}
\newcommand{\dx}{d_x}

\newcommand{\cred}{\textcolor{red}}

\newcommand{\UE}[1]{{\cred{UE: #1}}\xspace}

\newcommand{\pl}{\varrho_{l}}

\begin{document}

\maketitle

% Abstract
\begin{abstract}
Bayesian filtering is a key tool in many problems that involve the online processing of data, including data assimilation, optimal control, nonlinear tracking and others. Unfortunately, the implementation of filters for nonlinear, possibly high-dimensional, dynamical systems is far from straightforward, as computational methods have to meet a delicate trade-off involving stability, accuracy and computational cost. In this paper we investigate the design, and theoretical features, of constrained Bayesian filters for state space models. The constraint on the filter is given by a sequence of compact subsets of the state space which determine the sources and targets of the Markov transition kernels in the dynamical model. Subject to such constraints, we provide sufficient conditions for filter stability, as well as approximation error rates with respect to the original (unconstrained) Bayesian filter. Then, we look specifically into the implementation of constrained filters in a continuous-discrete setting where the state of the system is a continuous-time stochastic It\^o process but data are collected sequentially over a time grid. We propose an implementation of the constraint that relies on a data-driven modification of the drift of the It\^o process using barrier functions, and discuss the relation of this scheme with methods based on the Doob $h$-transform. Finally, we illustrate the theoretical results and the performance of the proposed methods in computer experiments for a partially-observed stochastic Lorenz 96 model.
\end{abstract}

% Keywords

% MSC codes

%%%%%%%%%%%%%%%%%%%%%%%%%%%%%%%%%%%%%%%
%
%%%%%%%%%%%%%%%%%%%%%%%%%%%%%%%%%%%%%%%
\section{Introduction}

%%%%%%%%%
%
%%%%%%%%%
\subsection{Motivation and background}

Bayesian filtering is a key tool in many problems that involve the online processing of data, including data assimilation \cite{Leeuwen09, Law15, Reich15, Fearnhead18}, optimal control \cite{Bensoussan92,Mitter02,Georgiou13}, nonlinear tracking \cite{Arulampalam02b,Ristic04,Stone13}, etc. The filtering problem is to compute the probability distribution of the state of a dynamical system conditional on available observations. The latter are collected over time, hence the distribution of the state has to be updated when new data becomes available \cite{Anderson79,Bain08}.  

The exact implementation of Bayesian filters is rarely feasible beyond the Kalman filter for linear and Gaussian models \cite{Kalman60,Kalman61,Anderson79,Ristic04}. In general, approximations are needed, which include Kalman-based methods \cite{Anderson79,Evensen03,Julier04}, recursive Monte Carlo algorithms \cite{Gordon93,Doucet00,Evensen03,Djuric03}, variational schemes \cite{Smidl08,Frazier23} and combinations of them. All these methods suffer from limitations which are related to a complex trade-off involving stability, accuracy and computational cost. For example, Kalman-based methods become very inaccurate when the system of interest displays strong nonlinearities or is subject to heavy-tailed perturbations \cite{Ristic04}, and particle filters (PFs) \cite{Gordon93,Kitagawa96,Doucet00,Djuric03} usually struggle when the observations are highly-informative or the model dimension is high \cite{Bengtsson08,Snyder08}.

One approach to the implementation of robust-yet-accurate filters is the imposition of constraints on the dynamics of the system of interest. The key idea underlying this strategy is to restrict the exploration of the state space to typical solutions, while excluding others which are feasible but very unlikely. In this way, one expects to obtain algorithms which are numerically stable and yield good accuracy with high probability. In a broad sense, auxiliary \cite{Pitt01,Douc09b,Branchini21} and twisted \cite{Whiteley14,Ala16} PFs can be interpreted as constrained algorithms, as they adapt the sampling kernel using the likelihood (auxiliary PFs) or optimize it to minimize the variance of certain estimators (twisted PFs). Another way to constrain the filter dynamics is to truncate the support of the system state, a solution that has been explored methodologically for Kalman filters \cite{Garcia12,Amor17} and also exploited to analyze filter stability \cite{Crisan20} and the convergence of PFs \cite{Heine08}.

In this work, we are specifically interested in continuous-discrete filtering problems, where the state of the system of interest is a continuous-time stochastic process $\X(t)$ characterized by an It\^o stochastic differential equation (SDE), and the observations $Y_1, Y_2, \ldots, Y_n, \ldots$  are collected sequentially at specific time instants $t_1, t_2, \ldots, t_n, \ldots $ \cite{Akyildiz24}. In this setting, any constraint on the system dynamics needs to be translated into a modification of the underlying SDE. Such modifications are not straightforward. Related methods include {\em guided} SDEs, where an additional drift term is constructed using Doob's $h$-transform \cite{Chopin23,Pieper-Sethmacher25,Pieper-Sethmacher25b,Pieper-Sethmacher25c,Eklund25}. This transform, denoted $h(x,t)$, is generally intractable, however, and it is common to resort to linearizations of the original SDE to approximate $h(x,t)$. The methods in \cite{Pieper-Sethmacher25,Pieper-Sethmacher25b,Pieper-Sethmacher25c,Eklund25} follow this approach. Moreover, they are developed for infinite-dimensional models and they enable sampling from regions of high probability given the observed data. The authors of \cite{Chopin23}, on the other hand, avoid linearizations and propose to approximate $h(x,t)$ using neural networks, which requires abundance of data and an offline training procedure. In any case, these techniques are computationally costly and remain subject to approximation errors.

%%%%%%%%%
%
%%%%%%%%%
\subsection{Contributions}

We investigate the problem of filtering for Markov state-space models (SSMs) \cite{Sarkka13} under constraints. In a Markov SSM, the state dynamics are described by a sequence of Markov transition kernels $K_1, K_2, \ldots, K_n, \ldots$ and we impose a constraint which is characterized by a sequence of compact subsets $\mC_0, \mC_1, \ldots, \mC_n, \ldots$ of the state space $\mX$. These sets specify the sources and targets of the kernels $K_n$, hence they yield a specific truncation of the state space $\mX$ for each $n$. In this setup, we look into the stability, approximation errors and numerical implementation of the resulting constrained Bayesian filters. 

We first prove that, under mild assumptions on the original (unconstrained) SSM, the constrained filter is exponentially stable. Stability is important because it enables the design of numerical filtering algorithms that do not accumulate errors over time \cite{DelMoral01c}. To analyze the approximation error rates of constrained filters, we impose some regularity assumptions both on the SSM and the constraint:
\begin{itemize}
\item[(a)] We parametrize the constraint by an index $l = 1, 2, \ldots$, in such a way that $\lim_{l\to\infty} \mC_n^l = \mX$ for every $n$. Intuitively, this means that the constraints on the state space are relaxed as $l$ increases, and we recover the original (unconstrained) model when $l\to\infty$.
\item[(b)] We assume that the compact subsets $\mC_0^l, \mC_1^l, \ldots, \mC_n^l, \ldots$ are sufficiently connected, in the sense that the probability mass $K_n(x,\mC_n^l)$, allocated to $\mC_n^l$ by the kernel $K_n$ with initial point $x' \in \mC_{n-1}^l$, is bounded away from zero by some constant $\varepsilon_n^l>0$ that depends on $l$. 
\end{itemize}
Under these assumptions, we find approximation error rates for the constrained filters with respect to $\varepsilon_n^l$. Moreover, when the positive bounds in assumption (b) above hold uniformly over time, i.e., $\varepsilon_n^l \ge \varepsilon^l >0$ for all $n$, the approximation rates can also be made uniform over $n$. Remarkably, this analysis suggests practical guidelines for the selection of the compact sets $\mC_0^l, \mC_1^l, \ldots, \mC_n^l, \ldots $ as superlevel sets of the likelihood functions in the unconstrained SSM.

In the continuous-discrete filtering setup, the observations are collected at times $t_1, t_2, \ldots, t_n, \ldots $, and the Markov kernels $K_n$ are determined by the SDE that characterizes the continuous-time state process $\X(t)$ for $t \in [t_{n-1}, t_n]$. We propose a simple (albeit possibly coarse) approximation of the constrained Markov kernels using barrier functions \cite{Margellos11,Clark21}. These functions can be efficiently computed online and they are used to modify the drift of the SDE to impose the constraint on the state space. As a result, we obtain a class of modified SDEs that can be easily simulated and built into sampling-based algorithms such as PFs and ensemble Kalman filters (EnKFs). We show some numerical results that illustrate the stability and accuracy of the proposed constrained filters on a stochastic Lorenz 96 model \cite{Grudzien20}

%%%%%%%%%
%
%%%%%%%%%
\subsection{Organization of the paper}

In Section \ref{sBackground} we introduce notation and provide some background on filter stability and constrained SSMs. Section \ref{sConstrained} is devoted to the analysis of constrained Bayesian filters, including their stability and approximation rates. %and convergence of constrained PFs. 
We note that Section \ref{sConstrained} is essentially theoretical. Readers interested only in methodological aspects may skip directly to Section \ref{sSDEs}, where we describe the continuous-discrete filtering setup, discuss Doob's $h$-transform for guided SDEs, and then introduce the proposed barrier-constrained SDEs. In Section \ref{sNumerics} we present numerical results for selected constrained Bayesian filters (using the barrier approach) and the stochastic Lorenz 96 model. Finally, Section \ref{sConclusions} is devoted to the conclusions.

%%%%%%%%%%%%%%%%%%%%%%%%%%%%%%%%%%%%%%%
%
%%%%%%%%%%%%%%%%%%%%%%%%%%%%%%%%%%%%%%%
\section{Background} \label{sBackground}

%%%%%%
%
%%%%%%
\subsection{State space models and Bayesian filtering} \label{ssIntroModels}

Let us construct a sequence of $\dx$-dimensional real random vectors $\{X_n\}_{n \ge 0}$ on the probability space $(\Omega,\mF,\mbP)$. The law of $X_0$ is denoted $\pi_0$ and the (random) dynamics of $X_n$ are determined by a sequence of Markov kernels 
\beq 
K_n : \mX \times \mB(\mX) \mapsto [0,1], \quad n >0,
\eeq  
where $\mX \subseteq \Real^{\dx}$ is the state space and $\mB(\mX)$ is the Borel $\sigma$-algebra of subsets of $\mX$. We also construct a sequence of uniformly bounded potential functions %\todo{Here potential used but also a mix of potential/likelihood. I think better now.}functions 
\beq
g_n : \mX \mapsto (0,\|g\|_\infty],
\eeq
where $\|g\|_\infty = \sup_{x\in\mX, n\in\mbN} |g_n(x)| < \infty$. 

\begin{definition}
If  $K=\{ K_n \}_{n \ge 1}$ is the sequence of Markov kernels and $g=\{g_n\}_{n\ge 1}$ is the sequence of uniformly bounded potentials, then we denote the system of interest as $\mS:=(\pi_0, K, g)$ and we refer to $\mS$ as a Markov {\em state space model} (SSM). 
\end{definition}

Let $\mP(\mX)$ denote the family of probability measures on $(\mX,\mB(\mX))$ and, for any real function $f : \mX \mapsto \Real$ and $\mu\in\mP(\mX)$, let $\mu(f)$ be shorthand for $\int_\mX f\sd\mu$. We introduce the sequence of operators
$
\Phi_n : \mP(\mX) \mapsto \mP(\mX)
$
such that
\beq
\Phi_n(\mu)(f) := \frac{
	K_n\mu(fg_n)
}{
	K_n\mu(g_n)
},
\label{eqPUop}
\eeq
where $K_n\mu(\sd z) := \int_{\mX} K_n(x,\sd z)\mu(\sd x)$. 
%\GJL{I do not follow exactly this expression.}
If we let $\pi_0$ be the probability law of the initial state $X_0$, then the operators $\{\Phi_n\}_{n\ge 1}$ yield the sequence of probability distributions
\beq
\pi_n := \Phi_n(\pi_{n-1}), \quad n\ge 1,
\eeq
generated by the model $\mS$. The composition of consecutive operators is denoted
\beq 
\Phi_{m:n} = \Phi_n \circ \Phi_{n-1} \circ \cdots \circ \Phi_{m+1},
\nn
\eeq
for $m\le n$, in such a way that $\pi_n = \Phi_{m:n}(\pi_m)$. We adopt the convention $\Phi_{n:n}(\pi)=\pi$. Also note that $\Phi_{n-1:n} = \Phi_n$.

In Bayesian filtering applications, the state sequence $\{X_n\}_{n\ge 0}$ is related to an observation sequence $\{Y_n\}_{n\ge 1}$, where $Y_n \in \mY \subseteq \mbR^{d_y}$, and $\mY$ is the $d_y$-dimensional observation space. In this setting, the potential $g_n(x)$ is the likelihood of the state $X_n=x$ given the observation $Y_n=y_n$, and $\pi_n$ is the conditional probability law of $X_n$ given $Y_1=y_1, \ldots, Y_n=y_n$, often referred to as the optimal (Bayesian) filter at time $n$. Also, the measure
\beq
\xi_n(\sd x) := K_n\pi_{n-1}(\sd x) = \int K_n(x',\sd x)\pi_{n-1}(\sd x')
\label{eq1ahead}
\eeq
is the conditional probability law of $X_n$ given $Y_1=y_1, \ldots, Y_{n-1}=y_{n-1}$, often referred to as the (one step ahead) predictive distribution of $X_n$. Bayes' theorem readily yields the relationship 
\beq
\pi_n(\sd x) = \frac{g_n(x) \xi_n(\sd x)}{\xi_n(g_n)}.
\label{eqBayesUpdate}
\eeq
Hereafter, we refer to $\Phi_n$ as the prediction-update (PU) operator, as it comprises the one-step-ahead prediction of \eqref{eq1ahead} and the Bayesian update of \eqref{eqBayesUpdate}. We also employ the term likelihood for $g_n$, rather than potential, to emphasize the relationship with the observations in the filtering problem.

%%%%%%
% 
%%%%%%
\subsection{A topology on the set of state space models}

Let $\mbS$ be the family of SSMs of the form $\mS=( \pi_0, K, g )$ for which the sequence $\{ \pi_n \}_{n\ge1}$ is well defined (in particular, we assume that $\xi_n(g_n) > 0$, for all $n \ge 1$, in Eq.~\eqref{eqBayesUpdate}).% in order to guarantee that $\Phi_n(\pi_{n-1})=\pi_n \in \mP(\mX)$ is well defined. %and, therefore, $\pi_n(f) < \infty$ whenever $f \in B(\mX) := \left\{ f:\mX\mapsto\Real ~\text{such that}~ \|f\|_\infty < \infty \right\}$.

For any two measures $\alpha,\beta\in\mP(\mX)$, let
\beq
\dtv{\alpha-\beta} := \sup_A |\alpha(A)-\beta(A)|,
\nn
\eeq
where the supremum is taken over all measurable subsets of $\mX$, denote the total variation norm of the difference $\alpha-\beta$. Following \cite{Crisan20}, we impose a topology $\mbD$ on the set of SSMs $\mbS$. To be specific, take some $\mS \in \mbS$ and let $\mS^l=(\pi_0^l,K^l,g^l)$, $l \ge 1$, be a sequence in $\mbS$, with $K^l=\{K_n^l\}_{n\ge 1}$ and $g^l=\{g_n^l\}_{l\ge 1}$. $\mS^l$ converges to $\mS$, and we denote 
$%\beq
\lim_{l\to\infty} \mS^l=\mS, 
%\nn
%\eeq
$
when
\beqa
\lim_{l\to\infty} \dtv{\pi_0-\pi_0^l} &=& 0, \nn\\
\lim_{l\to\infty} \dtv{K_n(x,\cdot)-K_n^l(x,\cdot)} &=& 0 ~~\text{for every $n \ge 1$, and}\nn \\
\lim_{l\to\infty} g_n^l(x) &=& g_n(x) ~~\text{for every $n\ge 1$ and $x\in\mX$.}\nn
\eeqa

Let $\{\Phi_n^l\}_{n\ge 1}$ denote the PU operators generated by model $\mS^l$, which in turn yield the sequence of measures $\pi_n^l = \Phi_{0:n}(\pi_0^l)$. The topology $\mbD$ has the property that convergence of the sequence of models $\mS^l \to \mS$ implies convergence of the probability laws $\pi_n^l \to \pi_n$ in total variation. Specifically, we have the following lemma taken from \cite{Crisan20}.

%%%
\begin{lemma} \label{lmMarginals} 
Let $\mS^l=(\pi_0^l,K^l,g^l)$, $l \ge 0$, and $\mS=(\pi_0,K,g)$ be elements of $\mbS$ with corresponding PU operators $\Phi_n^l$ and $\Phi_n$, $n \ge 1$, respectively. If $\lim_{l \to \infty} \mS^l=\mS$ in $\mbD$, then $\lim_{l \to \infty} \dtv{\Phi_{0:n}^l(\pi_0^l) - \Phi_{0:n}(\pi_0)}=0$.    
\end{lemma}

%%%%%%
%
%%%%%%
\subsection{Markov kernels and stability of the optimal filter}

Let $\{\Phi_n\}_{n \ge 0}$ be the sequence of PU operators generated by the Markov kernels $K=\{K_n\}_{n\ge 1}$ and the likelihoods $g=\{g_n\}_{n\ge 1}$ according to \eqref{eqPUop} (and note that $\Phi_n$ does {\em not} depend on $\pi_0$).

%\UE{not defined before, Prediction-Update}

%%%
\begin{definition}
The sequence of optimal filters $\{\pi_n\}_{n \ge 1}$, where $\pi_n=\Phi_n(\pi_{n-1})$, is stable when
$
\lim_{n\to\infty} \dtv{ \Phi_{0:n}(\mu) - \Phi_{0:n}(\eta) } = 0
$
for any pair of probability measures $\eta,\mu \in \mP(\mX)$.
\end{definition}
%%%

The terminology ``stability of the optimal filter'' is well established in the literature \cite{DelMoral01c,DelMoral04}. We note, however, that stability is a property of the operators $\Phi_n$, rather than the probability laws $\pi_n$ (see also \cite{Crisan20} for further details). A sufficient condition for $\Phi_n$ to be stable is that the Markov kernels in $K$ are {\em mixing} and the likelihood functions in $g$ integrable w.r.t. some reference probability density functions (pdfs) $u=\{u_n\}_{n\ge 1}$ on $\mX$.

%%%
\begin{definition} \label{defMixing}
The Markov kernels $K=\{K_n\}_{n\ge 1}$ are mixing with constant $0<\gamma<1$ when
\begin{itemize}
\item[(a)] $K_n(x',\sd x) = k_n(x',x) \sd x$ for every $n \ge 1$, where $k_n: \mX \times \mX \mapsto [0,\infty)$ is a pdf w.r.t. the Lebesgue measure $\sd x$, and
\item[(b)] there is a sequence of pdfs $u_n:\mX\mapsto [0,\infty)$ such that
\beq
\gamma u_n(x) \le k_n(x',x) \le \frac{1}{\gamma} u_n(x)
\quad \text{for every $n \ge 1$ and every $x' \in \mX$.}
\nn
\eeq
\end{itemize}
\end{definition}

\begin{lemma} \label{lmStability}
Let $\{ \Phi_n \}_{n \ge 1}$ be the PU operators generated by the Markov kernels $K=\{K_n\}_{n \ge 1}$ and the potential functions $g=\{g_n\}_{n\ge 1}$. Assume that the Markov kernels are mixing with constant $\gamma>0$ and pdfs $u=\{u_n\}_{n\ge 1}$, and
\beq
\int_\mX g_n(x)u_n(x) \sd x < \infty, \quad \text{for all $n\ge 1$}. 
\nn
\eeq
%\UE{ dx  in the integral ?} 
Then, for any $n > m \ge 0$, the composition of PU operators satisfies the inequality
\beq
\dtv{\Phi_{m:n}(\mu) - \Phi_{m:n}(\eta)} \le \frac{(1-\gamma^2)^{n-m}}{\gamma^2} \dtv{\mu - \eta}
\nn
\eeq
and, in particular, $\lim_{n\to\infty} \dtv{ \Phi_{0:n}(\mu) - \Phi_{0:n}(\eta) } = 0$ for any $\mu,\eta \in \mP(\mX)$.
\end{lemma}
%%%

This is a well known result (see Lemma 10 in \cite{Kunsch05}). The mixing condition in Definition \ref{defMixing} is sufficient but not necessary for stability.

\section{Filtering on constrained state space models} \label{sConstrained}

%\todo{I guess could include $l$ already in sections 3.1, 3.2 and say it is fixed. Only benefit it saves a little notation.}
%%%%%%
%
%%%%%%
\subsection{Constrained models} \label{ssConstrainedModels}

Let $\mS = (\pi_0, K, g)$ be a Markov SSM in $\mbS$ and choose a sequence $\mC=\{\mC_n\}_{n \ge 0}$ of compact subsets of the state space $\mX$. We %construct 
consider a constrained SSM $\hat \mS = (\hat \pi_0,\hat K,g) \in \mbS$, where $\hat K=\{\hat K_n\}_{n\ge 1}$,
\beq
\hat K_n(x',\sd x) := \frac{
    \Ind_{\mC_n}(x) K_n(x',\sd x)
}{
    \int_{\mC_n} K_n(x',\sd \tilde x)
},
\label{eqBasicConstraint}
\eeq
$\Ind_{\mC_n}$ is the indicator function for the set in the subscript, and $\hat\pi_0(\sd x) = \frac{\Ind_{\mC_0}(x)\pi_0(\sd x)}{\int_{\mC_0} \pi_0(\sd x)}$.
%, 
%\beq
%\hat K_n(x',\sd x) := \frac{\Ind_{\mC_n}(x) K_n(x',\sd x)}{K_n(x',\mC_n)},
%\nn
%\eeq
%$\Ind_{\mC_n}$ is the indicator function for the set in the subscript and $K_n(x',\mC_n) := \int \Ind_{\mC_n}(x) K_n(x',\sd x)$. 
It is easy to see that $\int_{\mC_n} \hat K_n(x',\sd x) = 1$ for any $x'\in\mX$ and hence the kernels in $\hat K$ are well defined. % whenever the original kernels in $K$ are well defined.

The PU operator $\hat\Phi_n$ for model $\hat\mS$ is constructed as in Eq. \eqref{eqPUop}, i.e., for a real, measurable test function $f:\mX\mapsto\Real$ we have 
\beq
\hat\Phi_n(\mu)(f) := \frac{
    \hat K_n \mu(fg_n)
}{
    \hat K_n\mu(g_n)
}
\label{eqCPUop}
\eeq
and, given an initial distribution $\hat\pi_0$, it yields the sequence of marginal posterior distributions $\hat\pi_n = \hat\Phi_n(\hat\pi_{n-1}) = \hat \Phi_{0:n}(\hat\pi_0)$, where $\hat\Phi_{0:n}=\hat\Phi_n \circ \cdots \circ \hat\Phi_1$. Also note that
\beq
\hat\pi_n(\sd x) = \hat\Phi_n(\hat\pi_{n-1})(\sd x)= \frac{
    g_n(x)\hat\xi_n(\sd x)
}{
    \hat\xi_n(g_n)
},
\nn
\eeq
where $\hat\xi_n(\sd x) = \hat K_n\hat\pi_{n-1}(\sd x)$ is the constrained one-step-ahead predictive measure.

%%%%%%
%
%%%%%%
\subsection{Stability}

Let us assume that the SSM of interest, $\mS$, satisfies some mild regularity conditions regarding the Markov kernels $K=\{K_n\}_{n\ge 1}$ and the likelihoods $g=\{g_n\}_{n\ge 1}$.

%%%
\begin{assumption} \label{assModel1}
The elements of model $\mS=(\pi_0,K,g)\in\mbS$ satisfy the conditions below:  
\begin{itemize}
\item there are density functions $k_n:\mX\times\mX\mapsto(0,\infty)$ such that $K_n(x',\sd x)=k_n(x',x)\sd x$ (in particular, $k_n(x',x)>0$ for any $x',x\in \mX$ and $n\ge 1$),
\item $\|k\|_\infty := \sup_{n,x',x} k_n(x',x) < \infty$, and 
\item $g_n>0$ for every $n\ge 1$ and $\|g\|_\infty := \sup_{n\ge 1} \|g_n\|_\infty < \infty$.
\end{itemize}
\end{assumption}
%%%

The sufficient conditions for stability in Lemma \ref{lmStability} hold most naturally for a constrained SSM $\hat \mS=(\hat \pi_0, \hat K, g)$ constructed from a model $\mS=(\pi_0,K,g)$ that satisfies Assumption \ref{assModel1}.

%%%
\begin{theorem} \label{thStabilitySl}
Let $\mS=(\pi_0,K,g) \in \mbS$ be a SSM for which Assumption \ref{assModel1} holds. If we choose a sequence $\mC=\{\mC_n\}_{n\ge 0}$ such that 
\beq
\inf_{n\ge 1} \inf_{(x',x) \in \mC_{n-1} \times \mC_n} k_n(x',x) \ge \bar\gamma
\label{eqlowerboundinC}
\eeq
for some $\bar\gamma>0$, then there is some constant $\gamma > 0$ such that
\beq
\dtv{\hat \Phi_{0:n}(\mu) - \hat\Phi_{0:n}(\eta)} \le \frac{(1-\gamma^2)^n}{\gamma^2}\dtv{\mu-\eta},
\label{eqStabHatS}
\eeq
where $\hat\Phi_n$ is the constrained PU operator in \eqref{eqCPUop}. In particular, 
$$
\lim_{n\to\infty} \dtv{\hat \Phi_{0:n}(\mu) - \hat\Phi_{0:n}(\eta)} = 0 
\quad \text{for any $\mu,\eta \in \mP(\mX)$}
$$ 
and the SSM $\hat \mS$ generates a stable sequence of marginal probability laws $\{\hat \pi_n\}_{n\ge 0}$.
\end{theorem}
%%%

\begin{proof}
If \eqref{eqlowerboundinC} holds then, for any $x_o \in \mC_{n-1}$, we have
\beq
0 < \frac{\bar\gamma}{\|k\|_\infty} < \frac{k_n(x',x)}{k_n(x_o,x)} < \frac{\| k \|_\infty}{\bar\gamma} < \infty, \quad \forall (x',x)\in\mC_{n-1}\times\mC_n
\label{eqbarkmix}
\eeq 
where the densities $k_n(x',x)$ and $k_n(x_o,x)$ are {\em not} normalized on $\mC_n$. Moreover, \eqref{eqlowerboundinC} and Assumption \ref{assModel1} also yield
\beq
\bar\gamma |\mC_n| \le \int_{\mC_n} k_n(x_o, x)\sd x \le \|k\|_\infty|\mC_n|, \quad \text{for any $x_o \in \mC_{n-1}$,}
\nn
\eeq
where $|\mC_n|=\int_{\mC_n} \sd x$. As a consequence,
\beq
\frac{\bar\gamma}{\|k\|_\infty} \le \frac{
    \int_{\mC_n} k_n(x',x)\sd x
}{
    \int_{\mC_n} k_n(x_o,x)\sd x
} \le \frac{\|k\|_\infty}{\bar \gamma}
\label{eqRatioBounds}
\eeq
and combining \eqref{eqbarkmix} with \eqref{eqRatioBounds} yields 
\beq
\frac{\bar\gamma^2}{\|k\|_\infty^2} < \frac{\hat k_n(x',x)}{u_n(x)} < \frac{\| k \|_\infty^2}{\bar\gamma^2} < \infty, \quad \forall (x',x)\in\mC_{n-1}\times\mC_n,
\label{eqbarkmix2}
\eeq 
where $\hat k_n(x',x) = \frac{k_n(x',x)}{\int_{\mC_n} K_n(x', \sd x)}$ and $u_n(x)=\frac{k_n(x_o,x)}{\int_{\mC_n} K_n(x_o, \sd x)}$ are normalized probability densities on $\mC_n$. The inequalities in \eqref{eqbarkmix2} imply that the constrained kernels $\hat K=\{\hat K_n\}_{n \ge 1}$ are mixing with constant $\gamma := \frac{\bar\gamma^2}{\|k\|_\infty^2}>0$ and the pdfs $u_n(x)$ satisfy $\int_{\mC_n} g_n(x) u_n(x)\sd x < \|g\|_\infty < \infty$.

%where we have assumed $\|k\|_\infty < \infty$. 
%Moreover, if we let $\varepsilon_n := \inf_{\bar x'\in\mC_{n-1}} \int_{\mC_n} K_n(\bar x', \sd \bar x)$ (and note that $\varepsilon_n>0$ because $\mC_{n-1}$ is compact and $k_n(x',x)>0$ by assumption) then 
%Since expression \eqref{eqbarkmix} holds for any $\tilde x \in \mC_{n-1}$, then, in particular, it is valid for $\tilde x = x_o$, and we readily arrive at
%\beq
%\frac{\bar\gamma}{\|k\|_\infty} v_o(x) < \hat k_n(x',x) < \frac{\| k \|_\infty}{\bar\gamma} v_o(x), \quad \forall (x',x)\in\mC_{n-1}\times\mC_n,
%\label{eqbarkmix2}
%\eeq
%where $\hat k_n(x',x) = \frac{k_n(x',x)}{\int_{\mC_n} K_n(x', \sd x)}$ and $v_o(x)=\frac{k_n(x_o,x)}{\int_{\mC_n} K_n(x_o, \sd x)}$. The inequalities in \eqref{eqbarkmix2} imply that the constrained kernels $\hat K=\{K_n\}_{n \ge 1}$ are mixing with constant $\gamma := \frac{\bar\gamma \varepsilon_n}{\|k\|_\infty}>0$.

%Moreover, $\int_{\mC_n} g_n(x)k_n(x_o,x)\sd x \le \|g\|_\infty < \infty$ for any $x_o \in \mC_{n-1}$ and every $n \ge 1$.

Therefore, we can readily apply Lemma \ref{lmStability} to obtain the desired inequality \eqref{eqStabHatS}. 
\end{proof}

If the state $X_n$ is a homogeneous process, i.e., $K_n=K$ for some Markov kernel $K:\mX\times\mB(\mX)\mapsto[0,1]$, and Assumption \ref{assModel1} holds, it is always possible to choose compact sets $\{ \mC_n \}_{n \ge 1}$ such that the uniform lower bound in \eqref{eqlowerboundinC} exists. Also note that, while Gaussian kernels are not mixing according to Definition \ref{defMixing}, they satisfy Assumption \ref{assModel1} and, therefore, they can be constrained to yield a stable sequence of posterior measures $\hat \pi_n$, $n \ge 0$.

%%%%%
%
%%%%%
\subsection{Approximation errors} \label{ssApproxErr}

We are interested in constructing constrained models $\hat \mS$ for which the posterior marginals $\hat \pi_n$ are close to the marginals $\pi_n$ generated by the original model $\mS$ under (reasonable) regularity assumptions. To this end, we adopt the framework of \cite{Crisan20}. In particular, consider the sequence $\{\mC^l\}_{l\ge 1}$ where each element $\mC^l = \{ \mC_n^l \}_{n \ge 0}$ is a sequence of compact sets that can be used to constrain the model $\mS$. Specifically, we construct the sequence of constrained models $\hat \mS^l=(\hat \pi_0^l,\hat K^l, g)$, where $\hat K^l = \{\hat K_n^l\}_{n\ge 1}$, 
$$
\hat K_n^l(x',\sd x) = \frac{
    \Ind_{\mC_n^l}(x)K_n(x',\sd x)
}{
    \int_{\mC_n^l} K_n(x',\sd x)
}
\quad
\text{and}
\quad
\hat \pi_0^l(\sd x) = \frac{\Ind_{\mC_0^l}(x)\pi_0(\sd x)}{\int_{\mC_0^l} \pi_0(\sd x)}.
$$
If $\mC^l$ is constructed in such a way that $\lim_{l\to\infty} \mC_n^l = \mX$ for every $n$, then it is apparent that $\lim_{l\to\infty} \hat \mS^l = \mS$ in the topology $\mbD$ and, by Lemma \ref{lmMarginals},
\beq
\lim_{l\to\infty} \hat \pi_n^l = \lim_{l\to\infty} \hat \Phi_{0:n}^l(\hat \pi_0^l) = \pi_n \quad \text{in total variation,}
\nn
\eeq
where the operators $\hat\Phi_n^l$, $n \ge 1$, are constructed in the obvious way given the kernel $\hat K_n^l$ and Eq. \eqref{eqCPUop}. We also denote $\hat\xi_n^l(\sd x)=\hat K_n^l\hat \pi_{n-1}^l(\sd x)$.

Not every constraint $\mC^l=\{\mC_n^l\}_{n\ge 0}$ on the state space $\mX$ leads to a ``reasonable'' approximation of the SSM $\mS$ and the Bayesian filters $\{\pi_n\}_{n\ge 1}$. In general, the subsets $\mC_{n}^l$ should be chosen in such a way that they capture a significant mass of the transition probability, i.e., that given a point $x'\in\mC_{n-1}^l$, the transition probability 
\beq
\mbP(X_n \in \mC_n^l| X_{n-1}=x') = \int_{\mC_n^l} K_n(x',\sd x)
\nn
\eeq
should be large enough. An intuitive interpretation of this condition is that the sequence of subsets $\{\mC_n^l\}_{n\ge 1}$ should be ``in agreement'' with the (unconstrained) dynamics of the state sequence $X_n$. This is made formal by Assumption \ref{assCl-1} below.

%%%
\begin{assumption} \label{assCl-1}
The family of compact subsets $\{\mC_n^l\}_{n\ge 0, l\ge 1}$ of the state space $\mX$ is constructed to guarantee that, for every $n\ge 0$, $\lim_{l\to\infty} \mC_n^l = \mX$ and there are positive sequences $\{ \varepsilon_n^l \}_{n \ge 0, l \ge 1}$, increasing with $l$ for every fixed $n$, such that
\beq
\pi_0(\mC_0^l) \ge \varepsilon_0^l, \quad \inf_{x'\in\mC_{n-1}^l} \int_{\mC_n^l} K_n(x',\sd x) \ge \varepsilon_n^l
\quad \text{and} \quad
\lim_{l\to\infty} \varepsilon_n^l = 1 ~~\forall n \ge 0.
\nn
\eeq
\end{assumption}
%%%

Based on the assumption above, and mild conditions on the original model $\mS$, we can obtain a first refinement of Lemma \ref{lmMarginals}.

%%%
\begin{theorem}\label{thConsistencyCl}
Let $\mS=(\pi_0,K,g)$ be a SSM where $g_n>0$ and $\| g_n \|_\infty < \infty$ for every $n \ge 1$. Let $\{\mC^l\}_{l\ge 1}$ be a sequence of constraints of the state space $\mX$ for which Assumption \ref{assCl-1} holds and let $\hat \mS^l=(\hat \pi_0^l,K^l,g)$, $l \ge 1$, be the sequence of constrained SSMs. For every measurable and real test function $f\in B(\mX)$ and every $n \ge 0$ there are finite constants $c_0, c_1, \ldots, c_n < \infty$, independent of $l$, such that
\beq
\left|
    \pi_n(f) - \hat \pi_n^l(f) 
\right| \le \sum_{i=0}^n c_i \frac{1-\varepsilon_i^l}{\varepsilon_i^l}.
\nn
\eeq
\end{theorem}
%%%

See Appendix \ref{apConsistencyCl} for a proof. The terms $\varepsilon_i^l$, $i \le n$, are the increasing sequences in Assumption \ref{assCl-1}, which satisfy $\lim_{l\to\infty} \varepsilon_i^l = 1$, hence $\left|
    \pi_n(f) - \hat \pi_n^l(f) 
\right| \to 0$ as expected from Lemma \ref{lmMarginals}.

Assumption \ref{assCl-1} can be strengthened by imposing that $\varepsilon_n^l \to 1$ uniformly over all $n \ge 0$.

%%%
\begin{assumption} \label{assCl-2}
The family of compact subsets $\{\mC_n^l\}_{n\ge 0, l\ge 1}$ of the state space $\mX$ is constructed to guarantee that, for every $n\ge 0$, $\lim_{l\to\infty} \mC_n^l = \mX$ and there is a single increasing sequence $\varepsilon^l > 0$ such that
\beq
\pi_0(\mC_0^l) \ge \varepsilon^l, \quad \inf_{x'\in\mC_{n-1}^l} \int_{\mC_n^l} K_n(x',\sd x) \ge \varepsilon^l
\quad \text{and} \quad
\lim_{l\to\infty} \varepsilon^l = 1.
\nn
\eeq
\end{assumption}
%%%

This stronger assumption on the constraints $\{\mC^l\}_{l \ge 1}$ yields a straightforward corollary of Theorem \ref{thConsistencyCl}.

%%%
\begin{corollary}\label{corConsistencyCl}
Let $\mS=(\pi_0,K,g)$ be a SSM where $g_n>0$ and $\| g_n \|_\infty < \infty$ for every $n \ge 1$. Let $\{\mC^l\}_{l\ge 1}$ be a sequence of constraints of the state space $\mX$ for which Assumption \ref{assCl-2} holds and let $\hat \mS^l=(\hat \pi_0^l,K^l,g)$, $l\ge 1$, be the sequence of constrained SSMs. For every measurable and real test function $f\in B(\mX)$ and every $n \ge 0$ there is a finite constant ${\rm c}_n<\infty$ such that
\beq
\left|
    \pi_n(f) - \hat \pi_n^l(f) 
\right| \le {\rm c}_n \frac{1-\varepsilon^l}{\varepsilon^l}.
\nn
\eeq
\end{corollary}
%%%

The proof is trivial: Assumption \ref{assCl-2} yields $\varepsilon_n^l = \varepsilon^l$ for all $n$, and ${\rm c}_n = \sum_{i \le n} c_i$, where the $c_i$'s are the constants in Theorem \ref{thConsistencyCl}.

The bounds on the approximation errors $|\pi_n(f) - \hat \pi_n^l(f)|$ can be improved based on Assumption \ref{assCl-2} and additional conditions on the base model $\mS$. Specifically, we may assume that $\mS$ generates a stable sequence of optimal filters $\pi_n = \Phi_{0:n}(\pi_0)$. Note that this is possible even if the Markov kernels $K=\{K_n\}_{n\ge 1}$ are not mixing according to Definition \ref{defMixing}. The following assumption identifies models that generate stable filters with arbitrary (possibly non-geometric) contraction rates.

%%%
\begin{assumption} \label{assModel2}
The model $\mS=(\pi_0,K,g) \in \mbS$ yields a stable sequence of optimal filters. In particular, there is a decreasing sequence $\{r_i\}_{i\ge 0}$ such that $r_0=1$, $\sum_{i=0}^\infty r_i < \infty$ and
$$
\dtv{\Phi_{m:n}(\mu)-\Phi_{m:n}(\eta)} \le r_{n-m}\dtv{\mu-\eta}
$$
for any pair of probability measures $\mu,\eta \in \mP(\mX)$ and $n \ge m$.
\end{assumption}
%%%

Under Assumption \ref{assModel2} it is possible to obtain uniform (over time $n$) bounds on the approximation error $|\pi_n(f) - \hat \pi_n^l(f)|$. However, the attainment of uniform approximation rates demands some additional conditions related to the likelihoods $g_n$. Specifically, we have worked under two alternative assumptions:
\begin{itemize}
\item[(a)] the likelihoods are uniformly bounded away from 0, i.e., $\inf_{n,x} g_n(x) > \zeta > 0$, or %as in the assumptions of Theorem \ref{thUniformPF} (or in Theorem 7.4.4 in \cite{DelMoral04}), or
\item[(b)] the likelihoods are continuous, vanishing as $\lim_{\|x\|\to\infty} g_n(x) = 0$, and the compact sets $\mC_n^l$ are superlevel sets of $g_n$ for every $n\ge 1$, i.e., $\mC_n^l = \{ x \in \mX: g_n(x) \ge \zeta_n^l \}$ for some real sequences $\zeta_n^l > 0$, with $\lim_{l\to\infty} \zeta_n^l = 0$ for all $n$.
\end{itemize}
Assumptions (a) and (b) above yield Theorem \ref{thUniformDTV-1} and Theorem \ref{thUniformDTV-2}, respectively.
%With the first assumption, we have Theorem \ref{thUniformDTV-1}, while the second assumption yields Theorem \ref{thUniformDTV-2}.

%\todo{Text "We emphasise that ...." could move to here.}

%%%
\begin{theorem} \label{thUniformDTV-1}
Let $\mS=(\pi_0,K,g)$ be a SSM for which Assumption \ref{assModel2} holds. Let $\{\mC^l\}_{l\ge 1}$ be a sequence of constraints of the state space $\mX$ that satisfies Assumption \ref{assCl-2}, and let $\hat \mS^l=(\hat \pi_0^l,K^l,g)$, $l \ge 1$, be the resulting sequence of constrained SSMs. If there are constants $\zeta>0$ and $\| g \|_\infty < \infty$ such that $\inf_{n\ge 1, x\in\mX} g_n(x) > \zeta$ and $\sup_{n \ge 1, x\in \mX} g_n(x) \le \| g \|_\infty$, respectively, then there is a finite constant $c<\infty$, independent of $n$, such that
\beq
\sup_{|f|\le 1} \left| \pi_n(f) - \hat \pi_n^l(f) \right| \le c \frac{1-\varepsilon^l}{\varepsilon^l}.
\label{eqThUnifDTV-1}
\eeq
In particular, $\lim_{l\to\infty} \sup_{n\ge 0} \dtv{\pi_n-\hat\pi_n^l}=0$.
\end{theorem}
%%%

See Appendix \ref{apUniformDTV-1} for a proof. Note that the supremum in \eqref{eqThUnifDTV-1} is taken over all real and measurable test functions $f:\mX\mapsto\Real$ such that $\|f\|_\infty \le 1$.

%%%
\begin{theorem} \label{thUniformDTV-2}
Let $\mS=(\pi_0,K,g)$ be a SSM for which Assumption \ref{assModel2} holds. Let $\{\mC^l\}_{l\ge 1}$ be a sequence of constraints of the state space $\mX$ that satisfies Assumption \ref{assCl-2}, and let $\hat \mS^l=(\hat \pi_0^l,K^l,g)$, $l \ge 1$, be the resulting sequence of constrained SSMs. If, in addition, the likelihoods $\{g_n\}_{n\ge 1}$ are continuous, $\lim_{\|x\|\to\infty} g_n(x)=0$ for all $n$ and every $\mC_n^l \subseteq \mX$ is a superlevel set of $g_n$, then 
\beq
\sup_{|f|\le 1} \left| \pi_n(f) - \hat \pi_n^l(f) \right| \le c \frac{1-\varepsilon^l}{\varepsilon^l}
\label{eqThUnifDTV-2}
\eeq
for some constant $c<\infty$ independent of $n$. In particular, $\lim_{l\to\infty} \sup_{n \ge 0} \dtv{\pi_n-\hat\pi_n^l}=0$.
\end{theorem}
%%%

See Appendix \ref{apUniformDTV-2} for a proof.

We emphasize that the (identical) inequalities \eqref{eqThUnifDTV-1} and \eqref{eqThUnifDTV-2} are obtained under different assumptions on the likelihoods $g_n$. In Theorem \ref{thUniformDTV-1}, we assume that the likelihood $g_n(x)$ is bounded away from zero for all $x \in \mX$, while in Theorem \ref{thUniformDTV-2}, however, the assumption is that $g_n(x) \to 0$ when $\| x \| \to \infty$. If the observations have the form $Y_n=m(X_n)+U_n$, where $m(\cdot)$ is an observation function and $U_n$ is Gaussian noise, then a bounded function $m(\cdot)$ yields $g_n > \zeta >0$ (and we can apply Theorem \ref{eqThUnifDTV-1}) whereas an unbounded $m(\cdot)$ typically yields the assumption in Theorem \ref{eqThUnifDTV-2}.

We also note that, for many models, Assumption \ref{assCl-2} (which is used in both Theorems \ref{eqThUnifDTV-1} and \ref{eqThUnifDTV-2}) can be hard to reconcile with the the requirement that the compact sets $\mC_n^l$ in the constraint be all of them superlevel sets of their respective likelihoods $g_n$. However, the latter condition is appealing because it yields a practical guideline for the construction of constraints and avoids the need for a uniform positive lower bound on the likelihoods (a condition that does not hold for many models). If we impose the weaker Assumption \ref{assCl-2} on the family of constraints $\{\mC^l\}_{l\ge 1}$, then the error bounds are not uniform over time any more, however we can still have the following result.

%%%
\begin{theorem} \label{thDTV-3}
Let $\mS=(\pi_0,K,g)$ be a SSM for which Assumption \ref{assModel2} holds. Let $\{\mC^l\}_{l\ge 1}$ be a sequence of constraints of the state space $\mX$ that satisfies Assumption \ref{assCl-1}, and let $\hat \mS^l=(\hat \pi_0^l,K^l,g)$, $l \ge 1$, be the resulting sequence of constrained SSMs. If, in addition, the likelihoods $\{g_n\}_{n\ge 1}$ are continuous, $\lim_{\|x\|\to\infty} g_n = 0$ for all $n$ and every $\mC_n^l \subseteq \mX$ is a superlevel set of $g_n$, then 
\beq
\sup_{|f|\le 1} \left| \pi_n(f) - \hat \pi_n^l(f) \right| \le 2 \sum_{i=0}^n r_{n-i} \frac{1-\varepsilon_i^l}{\varepsilon_i^l}.
\nn
\eeq
\end{theorem}
%%%

See Appendix \ref{apDTV-3} for a proof. 

%%%
\begin{remark} \label{rmDTV-3}
Note that $r_n \le r_0 = 1$, hence this bound is, in general, tighter than the bound in Theorem \ref{thConsistencyCl}, where the constants $c_i$ increase exponentially with index $i$. Moreover, if we restrict attention to a finite subsequence $\pi_0, \ldots, \pi_{n_0}$ for some finite $n_0$ then we still have a uniform bound $\dtv{\pi_n - \hat \pi_n^l} \le s_0\frac{1-\epsilon_0^l}{\epsilon_0^l}$, where $s_0 = 2\sum_{i=0}^{n_0} r_{n-i} < \infty$ and $\epsilon_0^l = \inf_{0 \le i \le n_0} \varepsilon_i^l$, where $\epsilon_0^l \to 1$ as $l\to\infty$.
\end{remark}
%%%

%%%%%%%%%%%%%%%%%%%%%%%%%%%%%%%%%%%%%%%
%
%%%%%%%%%%%%%%%%%%%%%%%%%%%%%%%%%%%%%%%
\section{Continuous-discrete filtering} \label{sSDEs}

%\begin{itemize}
%\item We have an SDE $dX=\GJL{a(X,t)}dt + \sigma(X,t)dW$
%\item From this SDE we want to implement a constrained kernel $\hat K_n(x',dx)$
%\item This is equivalent to implementing a SDE $dX=\GJL{a}(X,t)dt + \sigma(X,t)dW$ for $t \in [t_0,t_0+\Delta]$ subject to $X(t_0)=x'$ and $X(t_0+\Delta) \in C_{t_0+\Delta}$.
%\item Exact implementation: either (a) rejection sampling or (b) guided SDE using Doob's transform. Rejection sampling is very costly, Doob's not tractable.
%\item From Doob's transform, there is a hypertube that contains the paths from $t_0$ to $t_0+\Delta$ with high probability (level sets of the Doob transform)
%\item Inspired by that, we build a guided SDE using a barrier function. We need no to explain how we construct the hypertube and how we construct the barrier function.
%\item Write down the barrier function to fix notation:
%$$
%dX = \left(
 %   \GJL{a}(X,t) - \mu \sigma(X,t)\sigma(X,t)^\top \nabla \log c(X,t)
%\right) \sd t + \sigma(X,t)\sd W
%$$

%\end{itemize}

%%%%%%
%
%%%%%%
\subsection{State space model}

Let us assume that the signal of interest is a $\dx$-dimensional It\^o process $\X(t)$ in $(\Omega,\mF,\mbP)$, described by the stochastic differential equation (SDE) 
\beq\label{eq:sde}
\sd \X(t) = a(\X(t),t) \sd t + \sigma(\X(t),t) \sd \W(t),  
\quad t \in [0,T],
\eeq
where $t$ denotes continuous time, $T$ is some finite time horizon, $a:\mX \times [0,T] \mapsto \mX$ is the drift function, $\W(t)$ is a standard $d_w$-dimensional Wiener process and $\sigma:\mX \times [0,T] \mapsto \Real^{\dx \times d_w}$ is the diffusion (matrix) coefficient. For a given time grid 
\beq
0 = t_0 < t_1 < \cdots < t_n < \cdots < t_M = T,
\label{eqGrid}
\eeq
the discrete-time state sequence is $X_n := \X(t_n)$, $n \ge 0$, and the Markov kernels are given by %constructed as
\beq
K_n(x',\sd x) := \mbP\left( \X(t_n) \in \sd x | \X(t_{n-1})=x' \right), \quad \text{for $n \ge 1$}.
\nn
\eeq

One can generate approximate samples from the kernel $K_n(x',\sd x)$ using some numerical scheme for the SDE \eqref{eq:sde} on a possibly refined time grid 
\beq
t_{n,0}' = t_{n-1} < t_{n,1}' < \ldots < t_{n,j}' < \cdots < t_{n,J}' = t_n,
\nn
\eeq
with initial condition $X_{n-1} = x'$. For example, a standard Euler-Maruyama scheme yields
\beqa
\bar X_{n,0} &=& x', \quad \text{and}\nn\\
\bar X_{n,j} &=& \bar X_{n,j-1} + \Delta_{n,j} a(\bar X_{n,j-1},t_{n,j-1}') + \sqrt{\Delta_{n,j}} \sigma(\bar X_{n,j-1},t_{n,j-1}') Z_{n,j}, \quad \text{with~~} Z_{n,j} \sim \mN(0,I), \nn
\eeqa
for $j = 1, \ldots, J$ and $\Delta_{n,j} = t_{n,j}' - t_{n,j-1}'$. The last iterate, $\bar X_{n,J}$, is an approximate sample from the kernel $K_n(x',\sd x)$, hence $X_n \approx \bar X_{n,J}$.
%\todo{We do not mention rejection of 'large' solutions - perhaps standard for PF literature.}

The time grid \eqref{eqGrid} corresponds to the time instants where observations are collected and a likelihood function $g_n$ is available. We assume that the observations have the form \beq 
\label{eq: ObsModel}
Y_n = m(X_n) + U_n, \quad n=1, \ldots, M,
\eeq
where $U_n$ is a sequence of zero-mean, i.i.d., $\dy$-dimensional real r.v.s,  and $m: \mX \mapsto \mY$ is the measurement or observation function that maps the state space $\mX \subseteq \Real^{\dx}$ into the observation space $\mY \subseteq \Real^{\dy}$. The form of the likelihood $g_n(x)$ depends on the probability law of the noise variables $U_n$ and the observed value $Y_n= y_n$. For example, if the noise is zero-mean Gaussian with covariance matrix $\Sigma_u$, denoted $U_n \sim \mN(0,\Sigma_u)$, then we let
\beq
g_n(x) := \exp\left\{ 
    -\frac{1}{2}\left\|
        y_n - m(x)
    \right\|_{\Sigma_u^{-1}}^2
\right\}
\nn
\eeq
where $\left\| y_n - m(x) \right\|_{\Sigma_u^{-1}}^2 = \left(
    y_n - m(x)
\right)^\top \Sigma_u^{-1} \left(
    y_n - m(x)
\right)
$
and ${}^\top$ denotes transposition.

%%%%%%
%
%%%%%%
\subsection{Constrained kernels} \label{ssDoob}

Let $\mC = \{ \mC_n \}_{n \ge 0}$ be a constraint as described in Section \ref{ssConstrainedModels}. Sampling the kernel $\hat K_n(x',dx)$ in \eqref{eqBasicConstraint} is equivalent to simulating paths of the It\^o process $\X(t)$ that evolves from $\X(t_{n-1})=x'$ up to a point in the set $\mC_n$ at time $t_n$, i.e., we simulate $\X(t_n)$ conditional on the initial condition $\X(t_{n-1})=x'$ and the terminal condition $\X(t_n) \in \mC_n$.

One straightforward approach is to simulate trajectories of the original SDE \eqref{eq:sde} and then reject the samples $\X(t_n) \notin \mC_n$. However this can be extremely inefficient as the dimension $d_x$ increases. An alternative approach is to ``guide'' the It\^o process to generate samples contained in $\mC_n$ at time $t_n$. It is well known that the Doob $h$-transform \cite{Rogers00}
\beq
h(x,t) := \mbP\left( \X(t_n) \in \mC_n \mid \X(t)=x \right),
\nn
\eeq
which satisfies the backward Kolmogorov equation
\beq
\partial_t h + a \cdot \nabla h + \frac{1}{2}\mathrm{tr}(\sigma\sigma^\top\nabla^2 h) = 0, \qquad h(x,t_n)=\Ind_{\mC_n}(x),
\nn
\eeq
can be used for this purpose. To be specific, if we let $\sigma^2(\cdot,t):=\sigma(\cdot,t)\sigma(\cdot,t)^\top$ for conciseness, then the SDE
\beq
\sd \X^n(t) = \left[
    a(\X^n(t),t) + \sigma^2(\X^n(t),t) \nabla_x \log h(\X^n(t),t)
\right] \sd t + \sigma(\X^n(t),t)\sd \W(t),
\label{eqConstrainedSDE}
\eeq
with $\X^n(t_{n-1})=x'$, describes a modified It\^o process $\X^n(t)$ for $t \in [t_{n-1}, t_n]$, and the law of $\X^n(t)$ coincides with the law of $\X(t)$ conditional on $\X(t_{n-1})=x'$ and $\X(t_n)\in\mC_n$. Unfortunately, $h(X,t)$ is  intractable except for a few particular cases (and also expensive to approximate). 

Let
\beq
S_n(t) := \left\{x\in\mX : h(x,t) \ge \exp\{-r(t)\} \right\},
\quad t_{n-1} \le t \le t_n,
\nn
\eeq
be a flow of superlevel sets of the Doob $h$-transform generated by some function $r:[t_{n-1},t_n] \mapsto (0,\infty)$, such that $\lim_{t\to t_n} r(t) = \infty$ to ensure $S_n(t_n) = \mC_n$. Given $r(t)$, there is some $\epsilon_r \ge 0$ such that
%\beq
%\int_{S_n(t)} \sff_t(x) \sd x \ge 1-\epsilon_r, 
%\quad \forall t \in [t_{n-1}, t_n],
%\nn
%\eeq
%where $\sff_t(x)$ is the pdf of $\X^n(t)$, or, equivalently, 
\beq
\mbP\left( \X^n(t) \in S_n(t) \right) \ge 1-\epsilon_r, \quad \forall t \in [t_{n-1}, t_n].
\nn
\eeq
Assuming $\epsilon_r<<1$, the deterministic flow $S_n(t)$ captures most of the probability mass of the constrained process $\X^n(t)$ in Eq. \eqref{eqConstrainedSDE}. Intuitively, we can interpret it as a ``high probability tube'' connecting the initial condition $\X(t_{n-1})=x'$ with the terminal set $\mC_n$. % specified by the constraint $\mC$.

%%%%%%
%
%%%%%%
\subsection{Barrier-constrained kernels} \label{ssBarrier}

While approximating the Doob $h$-transform $h(x,t)$ is possible (several recent papers \cite{Chopin23,Pieper-Sethmacher25,Pieper-Sethmacher25b,Pieper-Sethmacher25c,Eklund25} have tackled this computation, e.g., using local linearization techniques or neural networks), we explore a simpler alternative where we construct a modified SDE that mimics the structure of \eqref{eqConstrainedSDE} but replaces $-\log h(x,t)$ by a simple-to-evaluate, differentiable ``barrier function'' $\log b(x,t)$. This barrier is designed to 
\begin{itemize}
\item penalize the continuous-time state trajectories that depart from a high probability tube denoted $\hat S_n(t) \subset \mX$, and
\item leave trajectories $\X^n(t) \in \hat S_n(t)$ unaltered. 
\end{itemize}
The tube $\hat S_n(t)$ is the deterministic flow of subsets of $\mX$ constructed by interpolating the constraint sets $\mC_{n-1}$ and $\mC_n$. In particular, $\hat S_n(t)$ is defined for $t \in [t_{n-1},t_n]$, with initial value $\hat S_n(t_{n-1}) = \mC_{n-1}$ and terminal value $\hat S_n(t_n) = \mC_n$. Again, this construction mimics the flow $S_n(t)$ of high-probability sets described in Section \ref{ssDoob} above. Specifically, the proposed barrier-constrained SDE takes the form
\beq
\sd \X^n(t) = \left[
    a(\X^n(t),t) - \mu \sigma^2(\X^n(t),t) \nabla_x \log b(\X^n(t),t)
\right] \sd t + \sigma(\X^n(t),t)\sd \W(t),
\label{eqBarrierSDE}
\eeq
where $\X^n(t_{n-1})=x'$ and $\mu>0$ is a parameter that controls the strength of the penalization when $\X^n(t) \notin \hat S_n(t)$. 

Let $\mC_n = \{ x\in\mX: \log g_n(x) \ge \log \zeta_n \}$ be a superlevel set of the likelihood $g_n$, in agreement with the assumptions of Theorem \ref{thUniformDTV-2}. This construction implies that the sets $\hat S_n(t)$ depend on the log-likelihood $g_n$, hence it is useful to make the barrier depend on this quantity as well. This is relatively simple and can be done in a number of ways. In this paper we use a shifted soft-plus function with parameters $\varrho$ and $\kappa \ge 1$,
\beq
u_{\varrho,\kappa}(z) := \frac{1}{\kappa} \log\left(
    1 + \exp\{\kappa(z-\varrho)\}
\right).
\nn
\eeq
Note that, for large $\kappa$, $u_{\varrho,\kappa}(z) \approx 0$, when $z<\varrho$, and $u_{\varrho,\kappa}(z) \approx z-\varrho$, when $z>\varrho$. The barrier function is then constructed as
\beq
\log b(x,t) := u_{\varrho,\kappa}( -\log g^n(x,t) )
\label{eqBarrierFunction}
\eeq
where $g^n(x,t)$ is an interpolated likelihood at time $t \in [t_{n-1},t_n]$, satisfying $g^n(x,t_{n-1}) = g_{n-1}(x)$ and $g^n(x,t_n)=g_n(x)$.

In Section \ref{sNumerics} below, we illustrate the construction of barrier functions for a Gaussian observation model, and study the numerical performance of some Bayesian filters that rely on the barrier-constrained SDE \eqref{eqBarrierSDE} to approximately simulate from the constrained Markov kernel $\hat K_n(x',\sd x)$.

\section{Numerical examples} \label{sNumerics}

%%%%%%
%
%%%%%%
\subsection{Stochastic Lorenz 96 model} 

To illustrate numerically both the analysis in Section \ref{sConstrained} and the methodology described in Section \ref{sSDEs}, we consider the $\dx$-dimensional stochastic Lorenz 96 model with additive noise \cite{Perez-Vieites18,Grudzien20} described by the system of SDEs
\beq
\sd \X_i(t) = \left[ \X_{i-1}(t) \left(  \X_{i+1}(t) - \X_{i-2}(t)  \right) - \X_i(t) + F\right] dt + \sigma_x  \sd \W_i(t), \quad i=0, \ldots, \dx-1,
\label{eqL96}
\eeq
with periodic boundaries given by 
\beq
\X_{-2}(t) = \X_{\dx-2}(t),
\quad \X_{-1}(t) = \X_{\dx-1}(t),
\quad  \text{and} \quad  \X_0(t) = \X_{\dx}(t), 
\nn
\eeq
where the forcing constant is $F=8$ (which yields turbulent dynamics), $\W_i(t)$, $i=0, \ldots, \dx-1$, are standard Wiener processes and $\sigma_x>0$ is a constant. The complete state is the $\dx \times 1$ vector $\X(t)=\left[ \X_0(t), \ldots, \X_{\dx-1}(t) \right]^\top$. %We integrate this model numerically using a (simple) Euler-Maruyama scheme with fixed step size $\Delta = 10^{-3}$.

We assume a linear and Gaussian observation model, namely 
\beq
Y_n = H^\top X_n + U_n,
\label{eqObsEq}
\eeq
where $U_n \sim \mN(0,\sigma_y^2 I_{d_y})$, $H$ is a $\dx \times \dy$ matrix and $X_n=\X(t_n)$. For each independent simulation trial, we generate $H$ randomly as
\beq
H = \left[
 	e_{m_{1}}, e_{m_{2}}, \ldots, e_{m_{\dy}} 
 \right] + V 
\nn
\eeq
where $e_{m}$ is a $\dx \times 1$ vector of 0s with a single value of 1 in the $m$-th entry, $V$ is a $\dx \times \dy$ matrix whose entries are independent Gaussian random variates, namely $V_{i,j} \sim \mN(0,\sigma_v^2)$ and $\sigma_v = 5 \times 10^{-4}$. The indices are drawn randomly from the set $\{1, \ldots, \dx\}$, with equal probabilities and no replacement. Intuitively, the $i$-th observation $Y_i$, $i\in\{1, \ldots, \dy\}$, corresponds to a linear combination of all the state variables, with $X_{m_i}$ being dominant and all other variables causing (individually small) interference in the measurement.
 
We approximate the ground-truth states $X_n$ in our experiments below by generating a discrete-time sequence $\X_l \approx \X(t_l')$, where $t_l'= l \Delta$, by way of the standard Euler-Maruyama scheme with step size $\Delta=10^{-3}$. We assume that observations are collected every $\Delta_o=0.1$ continuous time units. Hence, $X_n = \X(n\Delta_o) \approx \X_{\frac{\Delta_o}{\Delta}n} = \X_{100n}$. The Euler-Maruyama method is used in the numerical implementation of all the algorithms in this section.

%%%%%%
%
%%%%%%
\subsection{Barrier function}

Assume arbitrary but fixed observations $Y_n = y_n$, $n \ge 1$. The likelihood function for the observation model \eqref{eqObsEq} is
\beq
g_n(x) = \exp\left\{
    -\frac{1}{2\sigma_y^2} \| y_n - H^\top x \|^2
\right\}
\nn
\eeq
where $Y_n=y_n$ is the observation at time $t_n=n\Delta_o$. The constraint consists of superlevel sets $\mC_n = \{ x \in \mX: \log g_n(x) \ge \log \zeta \}$, where $\zeta\ge 0$ is the level. The tube $\hat S_n(t)$ is constructed as
\beq
S_n(t) := \{ x \in \mX: \log g^n(x,t) \ge \log \zeta \}, \quad t \in [t_{n-1}, t_n],
\nn
\eeq
where $t_n = n\Delta_o$ and the functions $g^n(x,t)$ are obtained by interpolating the observations $y_{n-1}$ and $y_n$. Specifically, we let 
\beq
y^n(t) := \frac{ 
    y_n-y_{n-1}
}{
    t_n-t_{n-1}
} t + \frac{
    y_{n-1}t_n - y_n t_{n-1}
}{
    t_n-t_{n-1}
}
~~\text{and}~~
g^n(x,t) := \exp\left\{
    -\frac{1}{2\sigma_y^2} \| y^n(t) - H^\top x \|^2
\right\}.
\nn
\eeq
Then, by construction, $g^n(x,t_{n-1}) = g_{n-1}(x)$, $g^n(x,t_{n}) = g_{n}(x)$, $\hat S_n(t_{n-1}) = \mC_{n-1}$ and $\hat S_n(t_n)= \mC_n$. 

Moreover, we note that
\beq
\hat S_n(t) = \{ x\in\mX: \log g^n(x,t) \ge \log \zeta \} 
= \left\{ x\in \mX: \| y^n(t) - H^\top x \| \le \sqrt{-\frac{\log \zeta}{2\sigma_y^2}} \right\},
\nn
\eeq
hence, $\hat S_n(t)$ is the preimage (on the state space $\mX$) of the ball $\{ y \in \mY: \| y^n(t) - y \| \le R \}$, with radius $R=\sqrt{-\frac{\log \zeta}{2\sigma_y^2}}$ (on the observation space $\mY$). The barrier function in \eqref{eqBarrierFunction} becomes
\beqa
\log b(x,t) &=& u_{\varrho,\kappa}\left( \frac{1}{2\sigma_y^2} \| y^n(t) - H^\top x \|^2 \right) \nn\\
&=& \frac{1}{\kappa} \log\left(
    1 + \exp\left\{\kappa\left( \frac{1}{2\sigma_y^2} \| y^n(t) - H^\top x \|^2 -\varrho \right)\right\}
\right),
\nn
\eeqa
where $\varrho = -\log \zeta$, and
\beq
\nabla_x \log b(x,t) = \frac{1}{\sigma_y^2} s\left(
\kappa\left( \frac{1}{2\sigma_y^2} \| y^n(t) - H^\top x \|^2 -\varrho \right)
\right) H (H^\top x - y^n(t)),
\eeq
where $s(z)=\frac{1}{1+\exp\{-z\}}$ is the logistic function.

%%%%%%
%
%%%%%%
\subsection{Stability of constrained filters}

Let $\mS^A=(\pi_0^A,K,g)$ and $\mS^B=(\pi_0^B,K,g)$ be two SSMs that differ only on their initial distribution. To assess the stability of the PU operators $\Phi_n$ (identical for the two models, as $\Phi_n$ depends only on $K_n$ and $g_n$) we ideally need to evaluate the total variation distance $\dtv{\pi_n^A-\pi_n^B}$, where $\pi_n^A$ and $\pi_n^B$ are the filtering distributions at time $n$ generated by model $\mS^A$ and model $\mS^B$, respectively. Unfortunately, this distance is intractable, hence we compute a proxy for $\dtv{\pi_n^A-\pi_n^B}$ by applying a simple discretisation procedure.

Specifically, at the time $t_n$ of the $n$-th observation, let us construct the hypercube 
\beq
\mH_n = \{ x \in \mX : \| x - X_n \|_\infty \le r \},
\nn
\eeq
centered at the true state vector $X_n$, with semi-length $r > 0$. This hypercube is then uniformly partitioned into $L_{\rm sub} := {\left\lfloor \frac{r}{r_{\text{sub}} } \right\rfloor^{\dx}}$ smaller sub-hypercubes $\mH_{n,l}$, $l=1,\ldots,L_{\rm sub}$, each of semi-length $r_{\rm sub}$. We note that $\mH_n = \cup_{l=1}^{L_{\rm sub}} \mH_{n,l}$ and $\mH_{n,l} \cap \mH_{n,j} = \emptyset$ when $l \ne j$. Given two measures $\alpha_n, \beta_n \in \mP(\mX)$, their total variation distance is approximated as
\beq \label{eq: dtv_formula}
\dtv {\alpha_n-\beta_n}
\approx
\frac{1}{2} \sum_{i} \big| \alpha_n (\mathcal{H}_{n,i}) - \beta_n (\mathcal{H}_{n,i}) \big|.
\eeq

From the SSMs $\mS^A$ and $\mS^B$ we construct constrained models $\hat \mS^A$ and $\hat \mS^B$. The constrained kernels $\hat K_n$ are implemented using the barrier method of Section \ref{ssBarrier} with parameters $\varrho=2$, $\kappa = 1$ and $\mu=10$. We simulate the Lorenz 96 model with dimension $\dx=10$, with $\dy=6$ observations at every time $t_n$, and noise scale parameters $\sigma_x^2 = \frac{1}{2}$ and $\sigma_y^2=\frac{1}{4}$. The discretisation of the total variation distance is carried out with $r=6$ and $r_{\rm sub}=3$. The filters $\hat\pi_n^A$ and $\hat\pi_n^B$ are approximated via a standard PF on the constrained models using $N$ particles, and we denote the Monte Carlo approximate filters as $\hat\pi_n^{A,N}$ and $\hat\pi_n^{B,N}$, respectively.

Figure \ref{fig:twopanels} illustrates the exponential stability predicted by Theorem \ref{thStabilitySl}. Specifically, in Figure \ref{fig: Task1}, we display the (approximate) total variation distances $\dtv{\pi_n^{A,N_1} - \pi_n^{B,N_1}}$ and $\dtv{\pi_n^{A,N_2} - \pi_n^{B,N_2}}$, obtained with $N_1=10,240$ and $N_2=1,024,000$ particles, averaged over 128 independent simulation runs. The prior distributions for this experiment are almost mutually singular, namely $\pi_0^A = \mathcal{N}(X_0-1,\dfrac{1}{4}I_{dx})$ and $\pi_0^B = \mathcal{N}(X_0+1,4I_{dx})$. The total variation distances are plotted versus the observation index (discrete time $n$). We see that the distance between the constrained filters with different priors decays exponentially fast, until it hits a Monte Carlo error floor, which is lower as we increase the number of particles $N$.

In Figure \ref{fig: Decay} we show a numerical estimate of the bound in \eqref{eqStabHatS}, which is obtained by a nonlinear least squares fit of the constant $\gamma$ (namely, $\hat \gamma \approx 0.496$). We see that the bound holds until the total variation curve approaches the Monte Carlo error floor.

\begin{figure}[h]
\centering
\subfloat[][{}]{%
  \parbox[t]{0.48\textwidth}{\centering
    \includegraphics[height=0.58\linewidth,keepaspectratio]{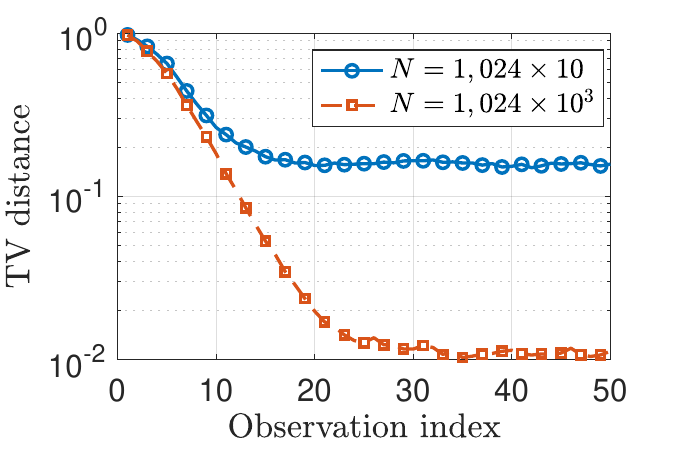}}%
  \label{fig: Task1}}
\hfill
\subfloat[][{}]{%
  \parbox[t]{0.48\textwidth}{\centering
    \includegraphics[height=0.58\linewidth,keepaspectratio]{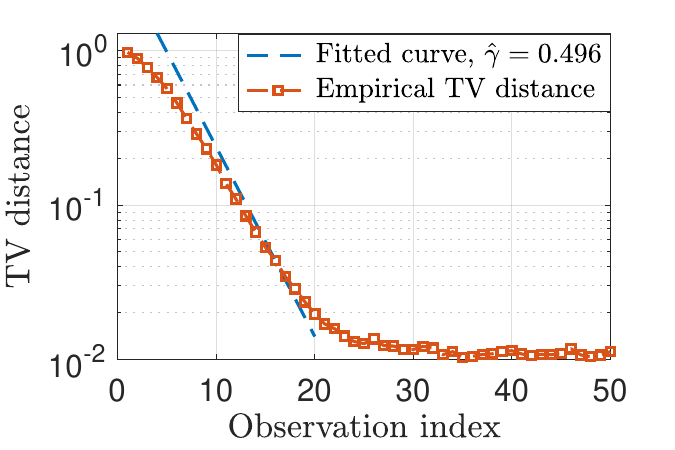}}%
  \label{fig: Decay}}
\caption{(a) The distances $\dtv{\hat\pi_n^{A,N_1} - \hat\pi_n^{B,N_1}}$ and $\dtv{\hat\pi_n^{A,N_2} - \hat\pi_n^{B,N_2}}$ versus observation index $n$, with number of particles $N_1=10,240$ and $N_2=1,024,000$ and barrier parameter $\varrho=2$. 
(b) The distance  $\dtv{\hat \pi_n^{A,N_2} - \hat \pi_n^{B,N_2}}$, barrier parameter $\varrho=2$ and the nonlinear fit $\frac{(1-\hat\gamma^2)^n}{\hat\gamma^2}\dtv{\hat\pi_0^A -\hat\pi_0^B}$, with $\hat\gamma = 0.496$, of the contraction bound in \eqref{eqStabHatS}.}
%(b) The distance  $\dtv{\hat \pi_n^{A,N} - \hat \pi_n^{B,N}}$ with  $N=1,024,000$, barrier parameter $\varrho=2$ and the nonlinear fit $\frac{(1-\hat\gamma^2)^n}{\hat\gamma^2}\dtv{\hat\pi_0^A -\hat\pi_0^B}$, with $\hat\gamma = 0.496$, of the contraction bound in \eqref{eqStabHatS}.}
\label{fig:twopanels}
\end{figure}

Next, let us take a sequence of three constraints $\mC^l$, $l=1, 2, 3$, which differ in the choice of the radius parameter $\varrho$ in the barrier function. In particular, we take  $\{\varrho_1, \varrho_2, \varrho_3 \} = \{0.5, ~~2, ~~4 \}$ and fix
 $\kappa = 1$ and $\mu=10$.
%In particular, we set $\kappa = 1$ and $\mu=10$ for all three constraints, while $\{\varrho_1, \varrho_2, \varrho_3 \} = \{0.5, ~~2, ~~4 \}$.
We then repeat the same experiment as in Figure \ref{fig:twopanels}, to obtain the decay of the total variation distance versus discrete time $n$.

The results are displayed in Figure \ref{fig: Task3}. It illustrates how the volume of the constraint sets $\mC_n^l$, determined by the barrier radius $\pl$ affects the stability  of the constrained filter. %For smaller 
As $\varrho_l$ decreases, the sets $\mC_n^l$ become smaller, leading to a stronger contraction of the posterior measures and, hence, a faster decay of the total variation distance. The fitted decay rates are  $\{\hat \gamma_1, \hat \gamma_2, \hat \gamma_3\} = \{0.518, 0.505, 0.487\}$, where larger $\gamma$ yields a faster decay.

%%%%%%
%
%%%%%%
\subsection{Volume of constraint sets}

%In this set of 
In these experiments we study the effect the volume of the constraint sets on the accuracy of the barrier-constrained filters. Similar to the previous experiments we take a stochastic Lorenz 96 model with state dimension $\dx=10$, $\dy=6$ observations at every time $t_n$, and noise scale parameters $\sigma_x^2 = \frac{1}{2}$ and $\sigma_y^2=\frac{1}{4}$. The barrier method is implemented with parameters $\kappa = 1$ and $\mu=10$, and radii $\{\varrho_1,\varrho_2,\varrho_3,\varrho_4\} = \{ 0.5, 1, 2, 4\}$. We note again that the radius $\varrho_l$ determines the volume of the sets $\mC_n^l$. With this setup, we run a standard PF with $N=50,000$ particles for reference, to approximate an ``optimal'' posterior mean estimator, a standard PF with just $N=500$ particles, and several barrier-constrained PFs (for $\varrho_1, \ldots, \varrho_4$) with $N=500$ particles.  

Figure \ref{fig: Task2} displays the results in terms of the normalized mean square error (NMSE) between between the posterior mean estimate $\hat x_n^N = \frac{1}{N}\sum_{i=1}^N x_n^i$ output by each filter and the true state $X_n$, namely ${\rm NMSE}_n = \frac{M\|X_n-\hat x_n^N\|^2}{\sum_{m=1}^M \| X_m \|^2}$, where $M$ is the number of observation times. The standard PF with $N=500$ can be expected to lose track of the system state after less than 20 discrete time steps (for $n < 20$) while the barrier-constrained filters perform close to the optimal posterior mean estimator for the complete simulation, with $50$ discrete time steps (the results are the same for longer simulations). The NMSE curves have been averaged over 128 independent simulation runs.

\begin{comment}

We consider Sequential Importance Resampling (SIR) filters with particle numbers  $N=50000$ and $N=500$, together with Barrier filters using $500$ particles and barrier parameter $\pl \in \{4,2,1,0.5 \}$. All filters propagate the same prior measure 
$\pi_0 \sim \mN(X_0^o, I_{\dx}) $. Figure \ref{fig: Task2} displays the averaged NMSE over 128 independent simulations as a function of the observation index. \UE{It clearly shows that the size of the constrained set has a significant influence on filter performance. When a large constraint radius by $\pl$
is used, the Barrier filter behaves similarly to the optimal (unconstrained) filter, here substituted by the SIR filter with $
N=50,000$ particles.  Importantly, all Barrier filters—despite using only 
$N=500$ particles—outperform the classical SIR filter with the same number of particles, which quickly deteriorates. This demonstrates that the Barrier mechanism provides substantial regularization, allowing the filter to achieve near-optimal performance with orders of magnitude fewer particles than a standard SIR.}

\end{comment}

\begin{figure}[h]
\centering
\subfloat[][{}]{%
  \parbox[t]{0.48\textwidth}{\centering
    \includegraphics[height=0.58\linewidth,keepaspectratio]{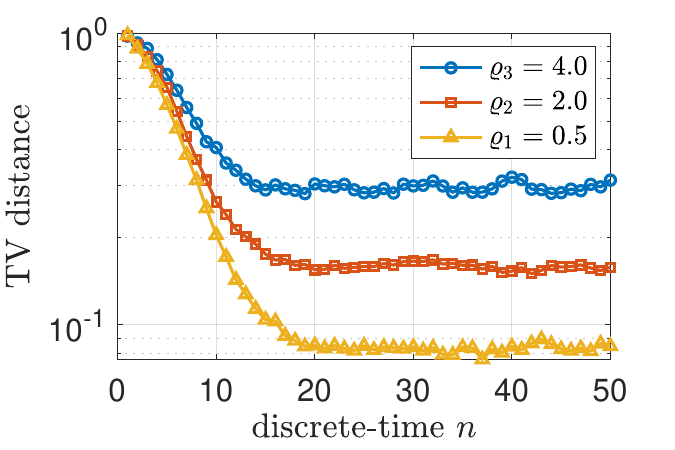}}
  \label{fig: Task3}}
\hfill
\subfloat[][{}]{%
  \parbox[t]{0.48\textwidth}{\centering
    \includegraphics[height=0.58\linewidth,keepaspectratio]{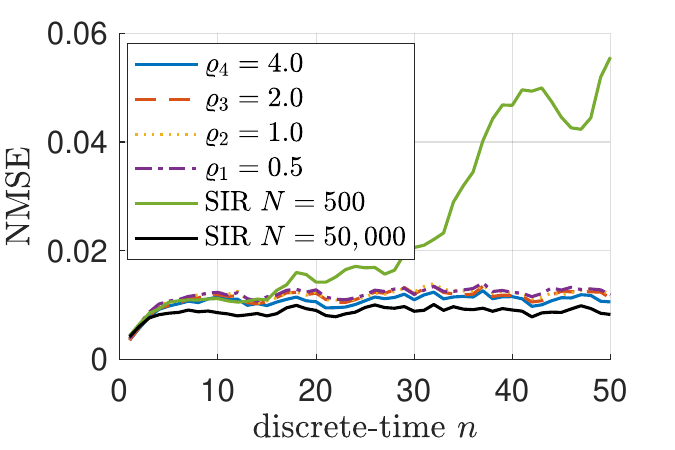}}
   \label{fig: Task2}}
\caption{(a) The (discretised) distance $\dtv{\pi_k^{A,N} - \pi_k^{B,N}}$ computed using standard PFs with $N=10,240$ particles and barrier parameters $\{\varrho_1, \varrho_2, \varrho_3\} = \{ 0.5, 2, 4 \}$. (b) NMSE versus observation index (discrete time $n$) for the standard PF (with $N=500$ and $N=50,00$ particles) and barrier-constrained PFs with different constraints parametrized by the barrier radii $\{\varrho_1,\varrho_2,\varrho_3,\varrho_4\} = \{0.5, 1, 2, 4\}$.}
%\label{fig:twopanels}
\end{figure}
 
%\todo{Figure 2 : swap colours on (b) for $\rho$ and label from $0$ or remove subscripts ?}

%%%%%%
%
%%%%%%
\subsection{Higher dimensional models}

In this final experiment we assess the accuracy of the barrier-constrained methods as the dimension of the state ($\dx$) and the observations ($\dy$) increase. Otherwise, the stochastic Lorenz 96 model parameters are $\sigma_x=\sigma_y=1$ and the continuous-time horizon is $T=10$, with observations collected every $\Delta_o=0.1$ continuous time units. For the barrier-based filters, we set $\mu=50$, $\kappa=100$ and $\varrho=4$. The schemes are rather robust to the choice of $\kappa$, however $\mu$ needs to be increased for moderate and large $d_x$.
%\todo[inline]{The parameters $\kappa$ and $\mu$ change. This needs to be discussed somewhere or at least mentioned. In conclusion ? My guess is a referee might ask.}
\begin{comment}

We use the normalized mean square error (NMSE) to compare the accuracy of the filters. To be explicit, let $X_k^o$ be the approximate `ground-truth' signal and  $\X_k$ be computed via a  filtering algorithm. Then, we define the NMSE at any observation time $t=t_k$ by
\beq
\text{NMSE}:= \frac{
	 \| \X_k^o - \X_k \|^2
}{
	 \| \X_k^o \|^2
}
\nn
\eeq

\end{comment}

To be specific, we compare the performance of 7 types of filters:
\begin{itemize}
\item a standard PF (labeled SIR),
\item an auxiliary particle filter (APF) \cite{Pitt01},
\item an ensemble Kalman filter (EnKF) \cite{Evensen03} implemented as in \cite{Perez-Vieites18},
\item a barrier-constrained standard PF (Barrier-SIR)
\item a barrier-constrained EnKF (Barrier-EnKF)
\item a guided PF (Guided) implemented as in \cite{Pieper-Sethmacher25b}, and
\item a barrier-constrained (Barrier-Guided) that incorporates a barrier term to the guided PF of \cite{Pieper-Sethmacher25b}.
\end{itemize}
For the guided PF we employ Algorithm 1 in \cite{Pieper-Sethmacher25b}, with the auxiliary drift chosen as $a(t)=0$. Under this choice, the simplified auxiliary dynamics reduces to a linear Ornstein–Uhlenbeck process without deterministic forcing. Although alternative choices of 
$a(t)$, such as frozen or filtered-mean linearizations, are possible and may further improve performance, we restrict attention here to the zero-drift auxiliary model in order to provide a clear baseline assessment of the method. % when combined constrained approach. 
The barrier-guided PF is constructed by combining the guidance mechanism with the barrier correction. Concretely, the proposal drift is obtained by augmenting the original Lorenz–96 drift with the guidance term and subsequently adding the barrier-induced correction term. The weight update remains that of Algorithm 1, so that the barrier modification acts only at the level of the proposal dynamics. This construction enables us to examine the effects of guidance and barrier stabilization in high-dimensional settings.

In Figure \ref{fig:MSE_fixed} we plot the NMSE attained by each filter versus the state dimension $\dx$ for a fixed number of particles (or ensemble size for the EnKFs) $N=750$. The observation dimension is $\dy = \lfloor 0.6\dx \rfloor$ and the observation matrix $H$ is fixed for each simulation but chosen randomly across different simulations. The results displayed are an average of 128 independent simulation runs.

We observe that the standard PFs break down quickly as the dimension increases and yields NMSEs over 1, which means that the power of the error is of the same order as the power of the signal $X_n$ itself.The guided PF improves upon the standard PF but still deteriorates as the dimension $d_x$ grows, with NMSEs remaining significantly above those of the barrier-based methods. The performance of the standard EnKF degrades quickly when $\dx > N$, while both the barrier-constrained PF and the barrier-constrained EnKF yield stable NMSEs as $\dx$ increases. The barrier-guided PF further improves the guided method and achieves consistently lower NMSEs than the simpler guided PF. The NMSE attained by the barrier-EnKF filter is around half the error of the barrier-SIR filter. 

Figure \ref{fig:MSE_varied} shows the results of a similar experiment where the number of particles for the different algorithms is increased linearly with the model dimension, namely $N =\lfloor 0.75  \dx \rfloor$. The simulation shows that the standard PFs still fail to track the stochastic Lorenz 96 model, while the barrier-based algorithms remain stable and improve their performance. The guided PF improves upon the standard PF under this scaling but still exhibits a gradual increase in NMSE as the dimension grows. The barrier-guided filter further improves the performance of the guided PF and maintains stable NMSE levels across dimensions. The barrier-EnKF filter yields an NMSE which is consistently lower than the NMSE of the standard EnKF for $\dx \ge 40$, with a very similar computational cost. 

Finally in Table~\ref{tab:largerdx} we illustrate that the three barrier-based methods we consider in Figure~\ref{fig:MSE_fixed} are applicable to larger dimensions. Here computations are made with a fixed sample size of $N=750$.

\begin{figure}[h]
\centering

\subfloat[][{}]{%
  \parbox[t]{0.49\textwidth}{\centering
    \includegraphics[width=0.8\linewidth,keepaspectratio]{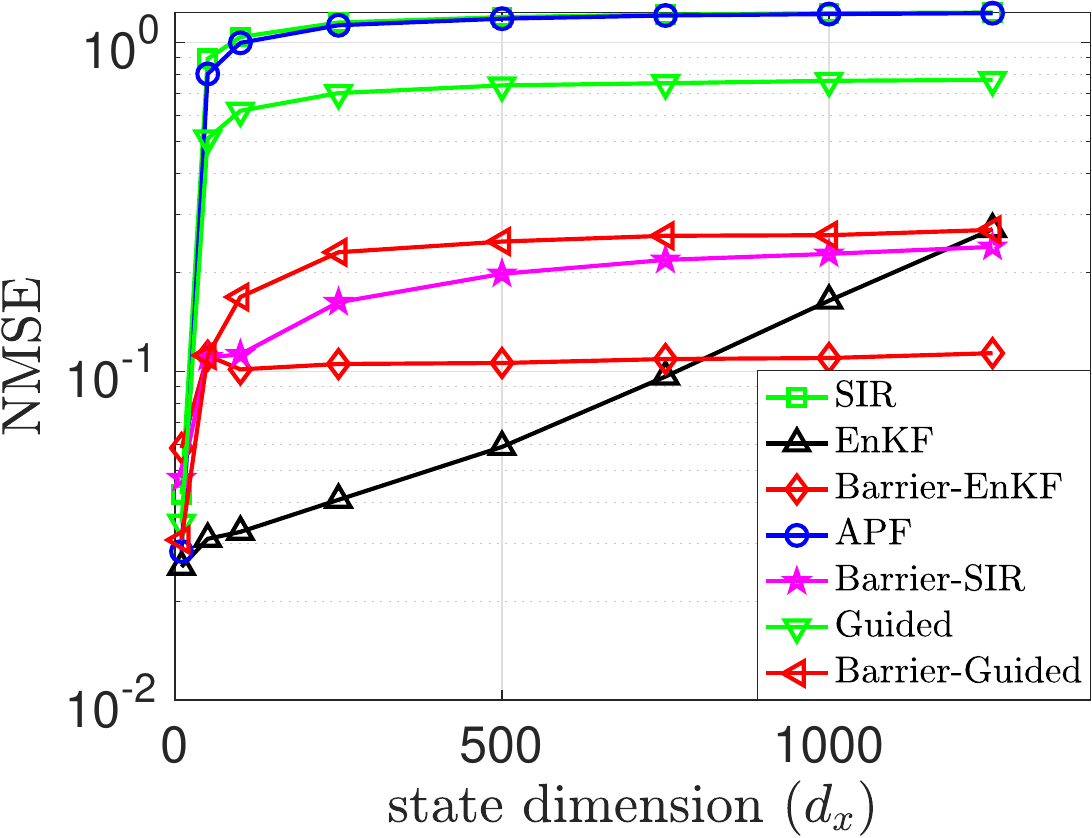}}%
  \label{fig:MSE_fixed}}
\hfill
\subfloat[][{}]{%
  \parbox[t]{0.49\textwidth}{\centering
    \includegraphics[width=0.8\linewidth,keepaspectratio]{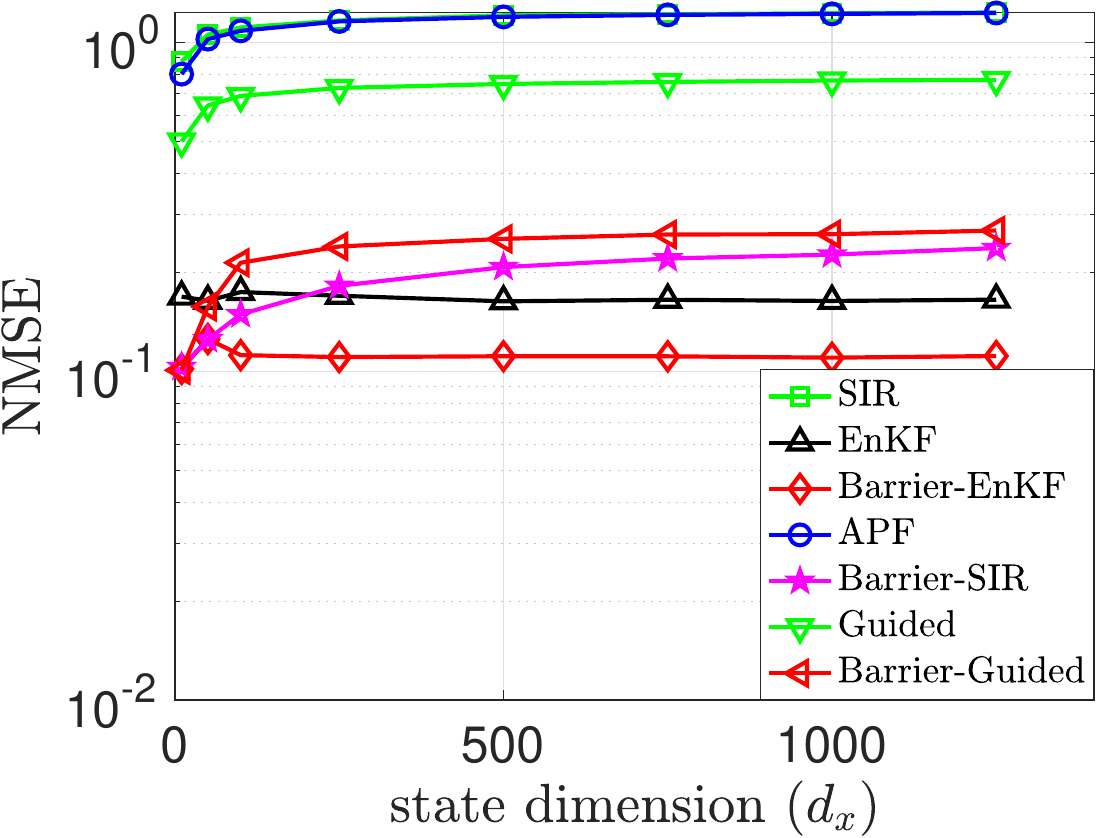}}%
  \label{fig:MSE_varied}}
\caption{NMSE versus state-space dimension $\dx$ for SIR, EnKF, APF, barrier-SIR, and barrier-EnKF algorithms. (a) All filters run with a fixed sample size $N=750$. (b) Sample size is scaled with dimension as $N=\lfloor 0.75 \dx \rfloor$.}
\label{fig:MSE}
\end{figure}

\begin{table}[]
\begin{center}
\begin{tabular}{ r | c c c }
    $d_x$       & 2500     & 5000     & 7500   \\ \hline \hline
 Barrier-SIR    & 0.2560  & 0.2705   & 0.1604    \\  
 Barrier-ENKF   & 0.1214  & 0.1485  & 0.2757    \\ 
 Barrier-Guided & 0.2770  & 0.2762  & 0.2797    
\end{tabular}
\end{center}
\caption{NMSE for the three Barrier based methods for $\dx=2500, 5000, 7500$. The results are averaged over four independent runs.}
\label{tab:largerdx}
\end{table}

%%%%%%%%%%%%%%%%%%%%%%%%%%%%%%%%%%%%%%%
%
%%%%%%%%%%%%%%%%%%%%%%%%%%%%%%%%%%%%%%%
\section{Conclusions} \label{sConclusions}

We have investigated a class of constrained Bayesian filters for state space models where the source and target of the Markov transition kernels are compact subsets of the state space constructed from the observed data. In this setup, we have analyzed the stability and approximation errors of Bayesian filters. To be specific, we have proved that, under mild assumptions on the original (unconstrained) state-space model, the constrained filter is exponentially stable. Stability is an important feature because it enables the design of numerical filters which do not accumulate error over time. We have also studied the approximation errors of constrained filters and obtained error rates that depend on a parameter of the constraint. Our analysis also provides guidelines on the selection of the sequence of compact sets that form the constraint, namely, that they should be selected as superlevel sets of the likelihood functions in the state space model. %The results of the approximation error analysis can also be leveraged to prove that particle filters run on the constrained state space model can be used to approximate the optimal (unconstrained) Bayesian filter under some regularity assumptions. Moreover, the selection of the compacts sets in the constraint as superlevel sets of the likelihood functions provides means to mitigate weight degeneracy in the particle filter.

In the second half of the paper we have turned our attention to the practical implementation of the proposed constrained filters in problems where data are collected over a discrete time grid but the signal of interest is a continuous-time stochastic It\^o process characterised by a SDE. We have proposed a relatively simple approximation of the resulting constrained Markov kernels using barrier functions to modify the drift of the original SDE. As a result, we have obtained a class of modified SDEs that can be easily simulated and built into sampling-based numerical filters, including particle filters and ensemble Kalman filters. Based on this approach, we have shown some numerical results that illustrate the stability and accuracy of the proposed constrained filters when tracking a (possibly high dimensional) stochastic Lorenz 96 system with partial and noisy observations.

%%%%%%%%%%%%%%%%%%%%%%%%%%%%%%%%%%%%%%%
%
%%%%%%%%%%%%%%%%%%%%%%%%%%%%%%%%%%%%%%%
\section*{Acknowledgments}

JM thanks the mathematical research institute MATRIX, in Australia, where part of this research was performed.

%%%%%%%%%%%%%%%%%%%%%%%%%%%%%%%%%%%%%%%
%
%%%%%%%%%%%%%%%%%%%%%%%%%%%%%%%%%%%%%%%
\appendix

%%%%%%
%
%%%%%%
\section{Proof of Theorem \ref{thConsistencyCl}} \label{apConsistencyCl}

We proceed by induction in time $n$. For $n=0$, we have the constrained prior $\hat\pi_0^l(\sd x)=\frac{\Ind_{\mC_0^l}(x)\pi_0(\sd x)}{\int_{\mC_0^l} \pi_0(\sd x)}$, hence the approximation error can be written as
\beqa
\left| \pi_0(f) - \hat\pi_0^l(f)\right| &=& \left|
    \int_\mX f(x) \left(
        1 - \frac{\Ind_{\mC_0^l}(x)}{\int_{\mC_0^l} \pi_0(\sd x)}
    \right) \pi_0(\sd x)
\right| \nn \\
&=& \left|
    \int_{\mC_0^l} f(x) \frac{\int_{\mC_0^l} \pi_0(\sd x) - 1}{\int_{\mC_0^l} \pi_0(\sd x)} \pi_0(\sd x)
    + \int_{\overline{\mC_0^l}} f(x) \pi_0(\sd x)
\right|, \nn
\nn
\eeqa
where the second equality follows from splitting the integral over the partition $\mX=\mC_0^l \cup \overline{\mC_0^l}$ (we denote $\overline{\mC_0^l}=\mX\backslash\mC_0^l$) and the fact that $\hat\pi_0^l(\Ind_S)=0$ for any $S\subseteq\overline{\mC_0^l}$. Assumption \ref{assCl-1} together with the equation above yield
\beq
\left| \pi_0(f) - \hat\pi_0^l(f)\right| \le \|f\|_\infty\left(
    \frac{1-\varepsilon_0^l}{\varepsilon_0^l} + (1-\varepsilon_0^l)
\right) \le c_0 \frac{1-\varepsilon_0^l}{\varepsilon_0^l},
\label{eqBase0}
\eeq
where $c_0 = 2\|f\|_\infty$ (note that $\varepsilon_0^l \le 1$). Inequality \eqref{eqBase0} is the base case of the induction. 

Let us assume that, at time $n-1$, with $n\ge 1$, 
\beq
\left|
    \pi_{n-1}(f)-\hat\pi_{n-1}^l(f)
\right| \le \sum_{i=0}^{n-1} \tilde c_i \frac{1-\varepsilon_i^l}{\varepsilon_i^l}
\label{eqInductionHypothesis}
\eeq
for some constants $\tilde c_i < \infty$ and the sequences $\varepsilon_i^l$ in Assumption \ref{assCl-1}, with $0 \le i \le n-1$.

At time $n$, we first look into the predictive measures $\xi_n=K_n\pi_{n-1}$ and $\xi_n^l=\hat K_n^l \hat\pi_{n-1}^l$. A straightforward triangle inequality yields
\beq
\left|
    \xi_n(f)-\hat\xi_n^l(f)
\right| \le \left|
    K_n\pi_{n-1}(f) - K_n\hat\pi_{n-1}^l(f)
\right| + \left|
    K_n\hat\pi_{n-1}^l(f) - \hat K_n^l\hat\pi_{n-1}^l(f)
\right|.
\label{eqInduc-0}
\eeq
If we let $K_n(f)(x'):=\int_\mX f(x)K_n(x',\sd x)$ then the first term on the right-hand side of \eqref{eqInduc-0} can be bounded as
\beq
\left|
    K_n\pi_{n-1}(f) - K_n\hat\pi_{n-1}^l(f)
\right| = \left|
    \pi_{n-1}(K_n(f)) - \hat\pi_{n-1}^l(K_n(f))
\right| \le \| f \|_\infty \sum_{i=0}^{n-1} \tilde c_i \frac{1-\varepsilon_i^l}{\varepsilon_i^l},
\label{eqInduc-1}
\eeq
where the inequality follows from the induction hypothesis \eqref{eqInductionHypothesis}. For the second term on the right-hand side of \eqref{eqInduc-0} a straightforward manipulation (similar to the base case) leads us to 
\beqa
\left|
    K_n\hat\pi_{n-1}^l(f) - \hat K_n^l\hat\pi_{n-1}^l(f)
\right| &\le& \| f \|_\infty \left[
    \int_{\mC_n^l} \int_{\mC_{n-1}^l} K_n(x',\sd x) \frac{
        1-\int_{\mC_n^l} K_n(x',\sd x)
    }{
        \int_{\mC_n^l} K_n(x',\sd \bar x)
    }
    \hat \pi_{n-1}^l(\sd x')
\right. \nn\\
&&+ \left.
    \int_{\overline{\mC_n^l}} \int_{\mC_{n-1}^l} K_n(x',\sd x)\hat \pi_{n-1}^l(\sd x')
\right]
\label{eqInduc-2}
\eeqa
and the inequality \eqref{eqInduc-2} together with Assumption \ref{assCl-1} readily yields
\beq
\left|
    K_n\pi_n^l(f) - \hat K_n^l\hat\pi_{n-1}^l(f)
\right| \le \| f \|_\infty \left[
    \frac{1-\varepsilon_n^l}{\varepsilon_n^l} + 1-\varepsilon_n^l
\right] \le 2\|f\|_\infty \frac{1-\varepsilon_n^l}{\varepsilon_n^l},
\label{eqInduc-3}
\eeq
where we have used the fact that $\int_{\mC_n^l} \int_{\mC_{n-1}^l} K_n(x',\sd x)\hat\pi_{n-1}^l(\sd x') \le 1$ and $\int_{\overline{\mC_n^l}} \int_{\mC_{n-1}^l} K_n(x',\sd x)\hat \pi_{n-1}^l(\sd x') = 1 - \int_{\mC_n^l} \int_{\mC_{n-1}^l} K_n(x',\sd x)\hat \pi_{n-1}^l(\sd x')$. Taking together \eqref{eqInduc-1} and \eqref{eqInduc-3}, we arrive at
\beq
\left|
    \xi_n(f)-\hat\xi_n^l(f)
\right| \le \sum_{i=0}^n \bar c_i \frac{1-\varepsilon_i^l}{\varepsilon_i^l},
\label{eqInduc-4}
\eeq
where $\bar c_i = \| f \|_\infty \tilde c_i$, for $i<n$, and $\bar c_n = 2\| f \|_\infty$.

To complete the proof, we note that Bayes' rule yields
\beq
\pi_n(f)-\hat\pi_n^l(f) = \frac{
    \xi_n(g_n f)
}{
    \xi_n(g_n)
} - \frac{
    \hat\xi_n^l(g_n f)
}{
    \hat\xi_n^l(g_n)
} \pm \frac{
    \hat\xi_n^l(g_n f)
}{
    \xi_n(g_n)
} %\nn\\
= \frac{
    \xi_n(g_n f) - \hat\xi_n^l(g_n f)
}{
    \xi_n(g_n)
} + \hat \pi_n^l(f) \frac{
    \hat\xi_n^l(g_n) - \xi_n(g_n)
}{
    \xi_n(g_n)
},
\label{eqBayesRule}
\eeq
which can be combined with the bound in \eqref{eqInduc-4} to obtain the inequality
\beq
\left|
    \pi_n(f)-\hat\pi_n^l(f)
\right| \le \sum_{i=0}^n c_i \frac{1-\varepsilon_i^l}{\varepsilon_i^l},
\text{~~where~~}
c_i = \frac{2\|f\|_\infty\|g_n\|_\infty \bar c_i}{\xi_n(g_n)}.
\nn
\eeq
We note that, by assumption, $g_n>0$ and $\|g_n\|_\infty$, hence $c_i<\infty$ for every $i=0, \ldots, n$. $\QED$

%%%%%%
%
%%%%%%
\section{Proof of Theorem \ref{thUniformDTV-1}} \label{apUniformDTV-1}

Let $\Phi_n$ be the PU operator associated to model $\mS$ and let $\Phi_{i:j}=\Phi_j \circ \cdots \circ \Phi_{i+1}$, with $\Phi_{i:i}(\pi) = \pi$ for any measure $\pi\in\mP(\mX)$. The approximation error can be written as
\beq
\hat\pi_n^l(f)-\pi_n(f) = \sum_{i=0}^{n-1} \Phi_{n-i:n}( \hat\pi_{n-i}^l )(f) - \Phi_{n-i:n}( \Phi_{n-i}(\hat \pi_{n-i-1}^l) )(f) + \Phi_{0:n}(\hat\pi_0^l)(f) - \Phi_{0:n}(\pi_0)(f),
\nn
\eeq
which readily yields the upper bounds\footnote{Note that for any probability measures $\alpha$ and $\beta$ one has $\dtv{\alpha-\beta} = \sup_{|f|\le 1} \frac{1}{2}|\alpha(f)-\beta(f)|$.}
\beqa
\sup_{|f|\le 1} \frac{1}{2}| \pi_n(f)-\hat\pi_n^l(f) | &\le&
\sum_{i=0}^{n-1} \dtv{ \Phi_{n-i:n}( \hat\pi_{n-i}^l ) - \Phi_{n-i:n}( \Phi_{n-i}(\hat \pi_{n-i-1}^l) ) }
\nn \\
&&+ \dtv{\Phi_{0:n}(\hat\pi_0^l) - \Phi_{0:n}(\pi_0)}
\nn \\
&\le& \sum_{i=0}^{n-1} r_i \dtv{ \hat\pi_{n-i}^l - \Phi_{n-i}(\hat \pi_{n-i-1}^l) } 
+ r_n\dtv{\pi_0-\hat\pi_0^l}
\nn\\
&=& \frac{1}{2} \left[
    \sum_{i=0}^{n-1} r_i \sup_{|f|\le 1} | \hat\pi_{n-i}^l(f) - \Phi_{n-i}(\hat \pi_{n-i-1}^l)(f) |
    + r_n |\pi_0(f) - \hat\pi_0^l(f) |
\right],
\label{eqAE-0}
\eeqa
where the second inequality follows from the stability Assumption \ref{assModel2}. 

For the last term on the right-hand side of \eqref{eqAE-0}, we obtain
\beq
\sup_{|f|\le 1} | \pi_0(f) - \hat\pi_0^l(f) | \le 2\frac{1-\varepsilon^l}{\varepsilon^l}
\label{eqAE-1}
\eeq
using Assumption \ref{assCl-2} and the same argument that leads to expression \eqref{eqBase0}. As for the $i$-th term in the sum of \eqref{eqAE-0}, applying the Bayes' rule in the same way as in Eq. \eqref{eqBayesRule} yields
\beqa
\sup_{|f|\le 1} | \hat\pi_{n-i}^l(f) - \Phi_{n-i}(\hat\pi_{n-i-1}^l)(f) | &=&
\sup_{|f|\le 1} \left|
    \frac{
        \hat K_{n-i}^l\hat\pi_{n-i-1}^l(g_{n-i}f) - K_{n-i}\hat\pi_{n-i-1}^l(g_{n-i}f)
    }{
        K_{n-i}\hat\pi_{n-i-1}^l(g_{n-i})
    } 
\right. \nn\\
&& \left.
    + \hat\pi_{n-i}^l(f)\frac{
        K_{n-i}\hat\pi_{n-i-1}^l(g_{n-i}) - \hat K_{n-i}^l\hat\pi_{n-i-1}^l(g_{n-i})
    }{
        K_{n-i}\hat\pi_{n-i-1}^l(g_{n-i})
    }
\right|
\nn \\
&\le& \frac{
    2 \int_{\mX} \int_{\mC_{n-i-1}^l} g_{n-i}(x)| K_{n-i}(x',\sd x) - \hat K_{n-i}^l(x',\sd x) | \hat\pi_{n-i-1}^l(\sd x')
}{
    K_{n-i}\hat\pi_{n-i-1}^l(g_{n-i})
}.
\label{eqAE-2}
\eeqa
In \eqref{eqAE-2}, the absolute difference between the kernels is
\beq
| K_{n-i}(x',\sd x) - \hat K_{n-i}^l(x',\sd x) | = K_{n-i}(x',\sd x) 
\quad \text{for} \quad \sd x \subset \overline{\mC_{n-i}^l},
\label{eqAE-3}
\eeq
and
\beq
| K_{n-i}(x',\sd x) - \hat K_{n-i}^l(x',\sd x) | = K_{n-1}(x',\sd x) \frac{
    1 - \int_{\mC_{n-i}^l} K_{n-i}(x',\sd \bar x)
}{
    \int_{\mC_{n-i}^l} K_{n-i}(x',\sd \bar x)
} \le K_{n-1}(x',\sd x) \frac{
    1 - \varepsilon^l
}{
    \varepsilon^l
}
\label{eqAE-4}
\eeq
for $\sd x \subset \mC_{n-i}^l$, with the inequality in \eqref{eqAE-3} following form Assumption \ref{assCl-2}. Hence, splitting the the first integral on the right-hand side of \eqref{eqAE-2} over the partition $\mX=\mC_{n-i}^l \cup \overline{\mC_{n-i}^l}$ and combining Eq. \eqref{eqAE-3} and inequality \eqref{eqAE-4} we arrive at
\beqa
\sup_{|f|\le 1} | \hat\pi_{n-i}^l(f) - \Phi_{n-i}(\hat\pi_{n-i-1}^l)(f) | &\le&
\frac{1-\varepsilon^l}{\varepsilon^l} \times  \frac{
    \int_{\mC_{n-i}^l} \int_{\mC_{n-i-1}^l} g_{n-i}(x)K_{n-1}(x',\sd x)\hat\pi_{n-i-1}^l(\sd x')
}{
    K_{n-i}\hat\pi_{n-i-1}^l(g_{n-i})
} \nn\\
&& + \frac{
    \int_{\overline{\mC_{n-i}^l}} \int_{\mC_{n-i-1}^l} g_{n-i}(x)K_{n-1}(x',\sd x)\hat\pi_{n-i-1}^l(\sd x')
}{
    K_{n-i}\hat\pi_{n-i-1}^l(g_{n-i})
} 
\nn\\
&=& \frac{1-\varepsilon^l}{\varepsilon^l} \times \frac{
    K_{n-i}\hat\pi_{n-i-1}^l(\Ind_{\mC_{n-i}^l} g_{n-i})
}{
    K_{n-i}\hat\pi_{n-i-1}^l(g_{n-i})
} + \frac{
    K_{n-i}\hat\pi_{n-i-1}^l(\Ind_{\overline{\mC_{n-i}^l}} g_{n-i})
}{
    K_{n-i}\hat\pi_{n-i-1}^l(g_{n-i})
}.
\label{eqAE-5}
\eeqa
However, by assumption, $\inf_i g_i \ge \zeta > 0$ and $\sup_i \| g_i \|_\infty \le \|g\|_\infty < \infty$, hence
\beq
\frac{
    K_{n-i}\hat\pi_{n-i-1}^l(\Ind_{\mC_{n-i}^l} g_{n-i})
}{
    K_{n-i}\hat\pi_{n-i-1}^l(g_{n-i})
} \le \frac{\|g\|_\infty}{\zeta} < \infty,
\label{eqAE-6}
\eeq
while
\beq
\frac{
    K_{n-i}\hat\pi_{n-i-1}^l(\Ind_{\overline{\mC_{n-i}^l}} g_{n-i})
}{
    K_{n-i}\hat\pi_{n-i-1}^l(g_{n-i})
} \le \frac{\|g\|_\infty}{\zeta} K_{n-i}\hat\pi_{n-i-1}^l(\Ind_{\overline{\mC_{n-i}^l}}) \le \frac{\|g\|_\infty}{\zeta} \times \frac{1-\varepsilon^l}{\varepsilon^l},
\label{eqAE-7}
\eeq
where the second inequality in \eqref{eqAE-7} follows from Assumption \ref{assCl-2} (and the fact that $\varepsilon^l\le 1$). 

Substituting \eqref{eqAE-6} and \eqref{eqAE-7} back into \eqref{eqAE-5} yields
\beq
\sup_{|f|\le 1} | \hat\pi_{n-i}^l(f) - \Phi_{n-i}(\hat\pi_{n-i-1}^l)(f) | \le
\frac{1-\varepsilon^l}{\varepsilon^l} \times \frac{2\|g\|_\infty}{\zeta}
\label{eqAE-8}
\eeq
and using the inequalities \eqref{eqAE-1} and \eqref{eqAE-8} in expression \eqref{eqAE-0} we obtain  
\beq
\sup_{|f|\le 1} | \pi_n(f)-\hat\pi_n^l(f) | \le 2 \frac{\|g\|_\infty}{\zeta} \times \frac{1-\varepsilon^l}{\varepsilon^l} \times \sum_{i=0}^n r_i.
\nn
\eeq
Since $\sum_{i=0}^\infty r_i < \infty$ by Assumption \ref{assModel2}, the inequality \eqref{eqThUnifDTV-1} in the statement of Theorem \ref{thUniformDTV-1} holds with constant $c = 2 \frac{\|g\|_\infty}{\zeta} \sum_{i=0}^\infty r_i < \infty$ independent of $n$ and $l$. $\QED$

%%%%%%
%
%%%%%%
\section{Proof of Theorem \ref{thUniformDTV-2}} \label{apUniformDTV-2}

First, we note that the subsets $\mC_n^l \subseteq \mX$ in the constraint $\mC^l$ have been defined to be compact. If every likelihood $g_n$ is continuous, with $\lim_{\|x\|\to\infty} g_n(x)=0$, then its superlevel sets are compact, and hence it is possible to choose the sets $\mC_n^l$ in the constraint $\mC^l$ as superlevel sets of $g_n$. 

We follow the same argument as in the proof of Theorem \ref{thUniformDTV-1} up to inequality \eqref{eqAE-5}. Then we note that 
\beq
\frac{
    K_{n-i}\hat\pi_{n-i-1}^l(\Ind_{\overline{\mC_{n-i}^l}} g_{n-i})
}{
    K_{n-i}\hat\pi_{n-i-1}^l(g_{n-i})
} \le \frac{
    \sup_{x\in\overline{\mC_{n-i}^l}} g_{n-i}(x)
}{
    \inf_{x\in\mC_{n-i}^l} g_{n-i}(x)
} \times \frac{
    K_{n-i}\hat\pi_{n-i-1}^l(\Ind_{\overline{\mC_{n-i}^l}})
}{
    K_{n-i}\hat\pi_{n-i-1}^l(\Ind_{\mC_{n-i}^l})
},
\label{eqAE2-0}
\eeq
and, moreover, 
\beq
\frac{
    \sup_{x\in\overline{\mC_{n-i}^l}} g_{n-i}(x)
}{
    \inf_{x\in\mC_{n-i}^l} g_{n-i}(x)
} \le 1 
\quad 
\text{and}
\quad
\frac{
    K_{n-i}\hat\pi_{n-i-1}^l(\Ind_{\overline{\mC_{n-i}^l}})
}{
    K_{n-i}\hat\pi_{n-i-1}^l(\Ind_{\mC_{n-i}^l})
} \le \frac{1-\varepsilon^l}{\varepsilon^l}
\label{eqAE2-1}
\eeq
because $\mC_{n-i}^l$ is a superlevel set of $g_{n-i}$ and Assumption \ref{assCl-2}, respectively. Substituting \eqref{eqAE2-1} into \eqref{eqAE2-0} and the result into \eqref{eqAE-5} yields
\beq
\sup_{|f|\le 1} | \hat\pi_{n-i}^l(f) - \Phi_{n-i}(\hat\pi_{n-i-1}^l)(f) | \le 2\frac{1-\varepsilon^l}{\varepsilon^l},
\label{eqAE2-2}
\eeq
where we have also used the fact that $K_{n-i}\hat\pi_{n-i-1}^l(\Ind_{\mC_{n-i}^l} g_{n-i}) \le K_{n-i}\hat\pi_{n-i-1}^l(g_{n-i})$. Combining \eqref{eqAE2-2}, \eqref{eqAE-1} and \eqref{eqAE-0} we arrive at
\beq
\sup_{|f|\le 1} | \pi_n(f) - \hat\pi_n^l(f) | \le
2 \frac{1-\varepsilon^l}{\varepsilon^l} \sum_{i=0}^n r_i.
\nn
\eeq
Since $\sum_{i=0}^\infty r_i<\infty$ by Assumption \ref{assModel2}, inequality \eqref{eqThUnifDTV-2} in the statement of Theorem \ref{thUniformDTV-2} holds with constant $c = 2\sum_{i=0}^\infty r_i<\infty$ independent of $n$ and $l$.

%%%%%%
%
%%%%%%
\section{Proof of Theorem \ref{thDTV-3}} \label{apDTV-3}

The argument is exactly the same as in the proof of Theorem \ref{thUniformDTV-2}, except that Assumption \ref{assCl-1} yields the time-dependent bounds
\beq
\sup_{|f|\le 1} |\pi_0(f)-\hat\pi_0^l(f)| \le 2\frac{1-\varepsilon_0^l}{\varepsilon_0^l}
\quad \text{and} \quad
\sup_{|f|\le 1} | \hat\pi_{n-i}^l(f) - \Phi_{n-i}(\hat\pi_{n-i-1}^l)(f) | \le 2\frac{1-\varepsilon_{n-i}^l}{\varepsilon_{n-i}^l}.
\label{eqAE3-0}
\eeq
Substituting \eqref{eqAE3-0} back into \eqref{eqAE-0} yields
$
\sup_{|f|\le 1} |\pi_n(f)-\hat\pi_n^l(f)| \le
2\sum_{i=0}^n \frac{1-\varepsilon_{n-i}^l}{\varepsilon_{n-i}^l} r_i. \quad \QED
$

%%%%%%%%%%%%%%%%%%%%%%%%%%%%%%%%%%%%%%%
%												
%%%%%%%%%%%%%%%%%%%%%%%%%%%%%%%%%%%%%%%
\bibliographystyle{siamplain}
\bibliography{bibliografia}

\end{document}